\documentclass[11pt]{article}
\usepackage{caption}

\usepackage[T1]{fontenc}
\usepackage[margin=1in]{geometry}
\usepackage{tikz-cd}
\usepackage{mathtools}
\usepackage{color}
\usepackage{proof}
\usepackage{caption}
\usepackage{bbm}
\usepackage{amssymb}
\usepackage{microtype}
\usepackage{xcolor}
\usepackage{enumitem}
\usepackage{amsthm}
\usepackage{thmtools,thm-restate}

\usepackage{amsmath}
\usepackage{amsfonts}
\usepackage{url}

\DeclareMathAlphabet{\mathpzc}{OT1}{pzc}{m}{it}

\usepackage[colorlinks,citecolor=blue]{hyperref}
\newtheorem{theorem}{Theorem}[section]
\newtheorem{lemma}[theorem]{Lemma}
\newtheorem{definition}[theorem]{Definition}

\newtheorem{corollary}[theorem]{Corollary}

\newtheorem{conjecture}{Conjecture} 

\newcounter{prob}
\newtheorem{problem}[prob]{Problem}

\newcommand{\eq}[1]{\hyperref[eq:#1]{(\ref*{eq:#1})}}
\renewcommand{\sec}[1]{\hyperref[sec:#1]{Section~\ref*{sec:#1}}}
\newcommand{\thm}[1]{\hyperref[thm:#1]{Theorem~\ref*{thm:#1}}}
\newcommand{\lem}[1]{\hyperref[lem:#1]{Lemma~\ref*{lem:#1}}}
\newcommand{\cor}[1]{\hyperref[cor:#1]{Corollary~\ref*{cor:#1}}}
\newcommand{\itm}[1]{\hyperref[itm:#1]{\ref*{itm:#1}}}
\newcommand{\app}[1]{\hyperref[app:#1]{Appendix~\ref*{app:#1}}}
\newcommand{\dfn}[1]{\hyperref[dfn:#1]{Definition~\ref*{dfn:#1}}}
\newcommand{\fig}[1]{\hyperref[fig:#1]{Figure~\ref*{fig:#1}}}
\newcommand{\clm}[1]{\hyperref[clm:#1]{Claim~\ref*{clm:#1}}}
\newcommand{\alg}[1]{\hyperref[alg:#1]{Algorithm~\ref*{alg:#1}}}
\newcommand{\stp}[1]{\hyperref[stp:#1]{Step~\ref*{stp:#1}}}
\newcommand{\asm}[1]{\hyperref[asm:#1]{Assumption~\ref*{asm:#1}}}
\newcommand{\prot}[1]{\hyperref[prot:#1]{Protocol~\ref*{prot:#1}}}
\newcommand{\prob}[1]{\hyperref[prob:#1]{Problem~\ref*{prob:#1}}}
\newcommand{\rmk}[1]{\hyperref[rmk:#1]{Remark~\ref*{rmk:#1}}}
\newcommand{\cons}[1]{\hyperref[cons:#1]{Construction~\ref*{cons:#1}}}
\newcommand{\conj}[1]{\hyperref[conj:#1]{Conjecture~\ref*{conj:#1}}}
\newcommand{\tbl}[1]{\hyperref[tbl:#1]{Table~\ref*{tbl:#1}}}

\usepackage{tabularx}

\let\originalleft\left
\let\originalright\right
\renewcommand{\left}{\mathopen{}\mathclose\bgroup\originalleft}
\renewcommand{\right}{\aftergroup\egroup\originalright}

\newcommand{\A}[0]{\mathcal{A}}
\newcommand{\B}[0]{\mathcal{B}}

\newcommand{\F}[0]{\mathcal{F}}
\newcommand{\G}[0]{\mathcal{G}}
\renewcommand{\H}[0]{\mathcal{H}}

\newcommand{\N}[0]{\mathcal{N}}
\renewcommand{\O}[0]{\mathcal{O}}

\DeclareMathOperator*{\Exp}{\mathbb{E}}
\DeclareMathOperator{\poly}{poly}

\newcommand{\Per}[0]{\mathrm{Per}}
\newcommand{\GPE}[0]{\mathrm{GPE}}
\newcommand{\PE}[0]{\mathrm{PE}}
\newcommand{\class}[1]{\mathsf{#1}}
\renewcommand{\P}[0]{\class{P}}
\newcommand{\sharpP}[0]{{\#\class{P}}}
\newcommand{\BPP}[0]{\class{BPP}}

\newcommand{\NP}[0]{\class{NP}}
\newcommand{\Mod}[0]{\class{Mod}}
\newcommand{\PH}[0]{\class{PH}}

\newcommand{\CAIP}[0]{\textsc{caip}}

\newcommand{\Id}[0]{\mathbbm{1}}

\newcommand{\bit}{\{0,1\}}

\usepackage{qcircuit}
\usepackage{environ}

\begin{document}
\title{The Computational Complexity of Quantum Determinants}

\author{
Shih-Han~Hung\footnote{Department of Computer Science, University of Texas at Austin. Email: shung@cs.utexas.edu.}
\and 
En-Jui Kuo\footnote{Joint Center for Quantum Information and Computer Science, NIST and University of Maryland, College Park,
Maryland 20742, USA. Email:kuoenjui@umd.edu}
}

\date{}

\maketitle

\begin{abstract}
In this work, we study the computational complexity of quantum determinants, a  $q$-deformation of matrix permanents: 
Given a complex number $q$ on the unit circle in the complex plane and an $n\times n$ matrix $X$, the $q$-permanent of $X$ is defined as 
$$\mathrm{Per}_q(X) = \sum_{\sigma\in S_n} q^{\ell(\sigma)}X_{1,\sigma(1)}\ldots X_{n,\sigma(n)},$$
where $\ell(\sigma)$ is the inversion number of permutation $\sigma$ in the symmetric group $S_n$ on $n$ elements. The function family generalizes determinant and permanent, which correspond to the cases $q=-1$ and $q=1$ respectively.

For worst-case hardness, by Liouville's approximation theorem and facts from algebraic number theory, we show that for primitive $m$-th root of unity $q$ for odd prime power $m=p^k$, exactly computing $q$-permanent is $\mathsf{Mod}_p\mathsf{P}$-hard. 
This implies that an efficient algorithm for computing $q$-permanent results in a collapse of the polynomial hierarchy. 
Next, we show that computing $q$-permanent can be achieved using an oracle that approximates to within a polynomial multiplicative error and a membership oracle for a finite set of algebraic integers. 
From this, an efficient approximation algorithm would also imply a collapse of the polynomial hierarchy. 
By random self-reducibility, computing $q$-permanent remains to be hard for a wide range of distributions satisfying a property called the strong autocorrelation property.
Specifically, this is proved via a reduction from $1$-permanent to $q$-permanent for $O(1/n^2)$ points $z$ on the unit circle. Since the family of permanent functions shares common algebraic structure, various techniques developed for the hardness of permanent can be generalized to $q$-permanents.
\end{abstract}

{
  \hypersetup{linkcolor=black}
}

\section{Introduction}

The study of matrix permanents has played an essential role in theoretical computer science and combinatorics. 
In a seminal work, Valiant proved that computing matrix permanent is $\sharpP$-hard, even if the matrix is over integers, non-negative integers or the set $\{0,1\}$ \cite{valiant1979completeness}.
This is considered to be one of the most influential results in computational complexity theory. 
These results have been used to justify of the hardness of other computational tasks in graph theory \cite{broder1986hard, curticapean2021parameterizing} and matrix analysis \cite{horn2012matrix}.
By contrast, matrix determinant is known to be polynomial-time computable from standard tools in elementary linear algebra \cite{mahajan1997determinant}. 

An interesting observation is that both permanent and determinant of a matrix $X$ can be written as a polynomial in matrix elements with phase coefficients, i.e., they are elements on the unit circle in the complex plane. 
More specifically, both functions can be written as a summation of $q^{\ell(\sigma)} X_{1,\sigma(1)}\ldots X_{n,\sigma(n)}$ over permutations $\sigma$ on $n$ elements for a phase $q$.
Here $\ell(\sigma)$ is the \emph{inversion number} of the permutation $\sigma$, defined as the number of transpositions on adjacent elements applied to restore $\sigma(1),\ldots,\sigma(n)$ to its natural order $1,2,\ldots,n$ \cite{yang1991q}. 
In the case of permanent, $q=1$, whereas in the case of determinant, $q=-1$. 
However, given the expression of both functions in a unified form, what makes their computational complexities so different remains elusive. 

A natural idea to explore this question is to study a $q$-deformation of the matrix polynomial, connecting the cases $q\in\{-1,1\}$. 
In mathematics literature, the $q$-deformation on the unit circle  or in $[-1,1]$ has been studied, and was given different names including $\mu$-permanent, $q$-permanent, $q$-determinant, or quantum determinant \cite{yang1991q,bapat1994inequalities,lal1998inequalities, da2018mu, da2010mu, de2018noncrossing, tagawa1993q}. 
However, to our knowledge, none of these results discussed the complexity of computing the analogue. 
In this work, we focus on the unit circle, and study the computational complexity of the function defined as
\begin{align}\label{eq:z-per}
  \Per_z(X) := \sum_{\sigma\in S_n} z^{\ell(\sigma)}  \prod_{i} X_{i,\sigma(i)}.
\end{align} 
for the symmetric group $S_n$ on $n$ elements and a complex number $z$ such that $|z|=1$. 
In the rest of the paper, for matrix $X$, we will call $\Per_z(X)$ the $z$-permanent of $X$. 

The $z$-permanent is not the only $q$-analogue of matrix permanents.
In particular, \emph{immanants} \cite{hartmann1985complexity,burgisser2000computational,mertens2011complexity,de2013determinant,spivak2020immanants,curticapean2021full} and \emph{fermionants} \cite{mertens2011complexity,bjorklund2019generalized,rugy2013determinant,de2013determinant} are two families of functions that have been widely studied. 
They are both defined with class functions, i.e., for some function $f$ invariant under conjugation on the symmetric group, these function can be written in the form
\begin{align}
    \sum_{\sigma\in S_n}f(\sigma) X_{1,\sigma(1)}\ldots X_{n,\sigma(n)}. 
\end{align}
For immanants, $f(\sigma)=\chi_\lambda(\sigma)$ where $\chi_\lambda$ is the character of the group representation $\lambda$ of $S_n$.
The permanent corresponds to the special case that $\chi_\lambda$ is the trivial representation, and the determinant corresponds to the sign representation.  
In a recent breakthrough, Curticapean classified the complexity of computing immanants for partition families.
The complexity of immanants is settled down by proving that under plausible assumptions, for every partition family $\Lambda$, matrix immanants are efficiently computable if and only if a quantity $b(\Lambda)$, which counts the boxes to the right of the first column in the Young diagram of $\lambda\in\Lambda$, is unbounded \cite{curticapean2021full}. 
For fermionants, $f(\sigma)=(-1)^n(-k)^{c(\sigma)}$, where $c(\sigma)$ denotes the number of cycles in $\sigma$. 
On the hardness side, Martins and Moore showed that fermionants are $\sharpP$-hard for $k>2$ under Turing reduction, and for $k=2$, the problem is $\oplus\mathsf{P}$-hard \cite{mertens2011complexity}. 
On the other hand, Bj\"{o}rklund, Kaski, and Williams gave an algorithm in time $2^{m-\Omega(m/\log\log q))}\cdot O(\mathrm{M}(q))$ for $m\times m$ matrices over a finite field $\mathbb F_q$ of $q$ elements, where $\mathrm{M}(q)$ is the time complexity of multiplication and division over $\mathbb F_q$ \cite{bjorklund2019generalized}.
These results justify that these functions are computationally intractable in general, except in special cases.

However, to our knowledge, no classification of the complexity of $z$-permanents has been given for any $z\notin\{-1,+1\}$, and none of the aforementioned results seems to obviously imply the hardness for this family, partly because the inversion number is not a class function.
In fact, no known complexity result for $q$-deformations of matrix permanents has even been given beyond class functions. 

In addition to the fact that the coefficients are not class functions, one may wonder why $z$-permanent (also called quantum determinant) is interesting object to study? 
First, in representation theory, the standard Lie algebra can be generalized into the so-called $q$-deformed algebra (also known as quantum group \cite{majid2000foundations}). Quantum determinants show up naturally in such contexts. 
More specifically, consider the $z$-Grassmann algebra which is the associative algebra $K[z]$ generated by $x_1, x_2, ..., x_n$ satisfying the $x_i^2=0$ for every $i\in[n]$, and $x_ix_j=zx_jx_i$ for $i<j$. 
It is straightforward to see that
\begin{equation}
    \left(\sum_{j=1}^{n}a_{1j}x_j\right)\left(\sum_{j=1}^{n}a_{2j}x_j\right)\ldots\left(\sum_{j=1}^{n}a_{nj}x_j\right)=\Per_{z}(A)x_1x_2\ldots x_n.
\end{equation}
which $A$ is $n$ by $n$ matrix such that $A_{ij}=a_{ij}$. 
Quantum determinants also serve as an important object in the representation of such $q$-deformed algebra \cite{etingof1998quantum}, and are closely related to the representation theory over different characteristics. 
Second, it is known that there are many applications of quantum determinants in various areas of physics, including string theory \cite{frohlich2006quantum} and conformal field theory \cite{gomez1996quantum, pasquier1990common}. 
In quantum information theory, quantum determinants also play an essential role in the Fock representations of $q$-deformed commutation relations \cite{bozejko2017fock} or anyonic relation \cite{goldin2004fock}, generalizing the standard bosonic/fermionic commutation relation \cite{bozejko1991example}.   

Motivated by the above observations, we initiate the study on the computational complexity of $z$-permanents, by giving general hardness results for complex number $z$ on the unit circle.
In the following sections, we give an overview of our contributions.

\subsection{Our Contributions}

In this paper, we study the computational complexity of $z$-permanent for an $m$-th primitive root of unity $z$ for integer $m>2$.
Our contribution is built upon tools from algebraic number theory. 
We summarize our result in~\tbl{main1}. From \sec{classical} to \sec{wa}, we prove one theorem in each section.

\begin{table}
\begin{center}
\begin{tabularx}{0.7\textwidth} { 
  | >{\raggedright\arraybackslash}X 
  | >{\centering\arraybackslash}X 
  | >{\raggedleft\arraybackslash}X | }
 \hline
 $ z=\zeta_{p^k}, \in {\displaystyle \mathbb{T} \setminus \{ \pm 1\},}
 |\Per_z(X)|^2$ & Exact & Approximate within multiplicative error $\poly(n)$ \\
 \hline
 Worst case &  \thm{classical hard} &  \thm{we to wa} \\
\hline
Average case for $X\sim  \H$  & constant fraction  \thm{random} &
\thm{GPE-z}
\\
\hline
\end{tabularx}
\end{center}
\caption{
We list hardness results into common four categories and specify the four main theorems shown in this paper.
}
\label{tbl:main1}
\end{table}

\subsubsection{Worst-case Hardness}
First, we study the hardness of computing $\zeta_m$-permanent of a $n\times n$ binary matrix $X$ for an $m$-th primitive root of unity $\zeta_m$.
By \eq{z-per}, we can write the $z$-permanent as a polynomial in $z$ whose coefficients can be written as multivariate polynomials in the elements of $X$.
For $z\notin\{+1,-1\}$, $z$-permanent is in general a complex number which can be written as an integral polynomial in $z$.
Since complex number arithmetic can only be computed using a finite precision, we first clarify what it means by an algorithm computing $\Per_z(X)$ exactly.

While we do not know how to give a sensible definition for every $z$ on the unit circle, applying facts of \emph{cyclotomic fields} yields a natural definition for a primitive root of unity and a binary matrix.
More specifically, $\zeta_m$-permanent of a binary matrix can be written as an integral polynomial in $\zeta_m$ of degree $d=\deg\Phi_m-1$, where $\Phi_m$ is the $m$-th cyclotomic polynomial.
This implies that representing $\Per_{\zeta_m}(X)$ as a degree-$d$ polynomial, the tuple of coefficients is unique.
Thus we define an algorithm which exactly computes $\Per_{\zeta_m}(X)$ to be one that outputs a set of coefficients $a_0,\ldots,a_d$ such that $\Per_{\zeta_m}(X)=\sum_{i=0}^d a_i\zeta_m^i$.

Given this definition, we show that for odd prime power $m=p^k$,  any oracle $\O$ that computes $\zeta_m$-permanent implies that $\BPP^{\Mod_p\P}\subseteq\BPP^\O$.
\begin{theorem}[Hardness of $z$-permanent, informal]\label{thm:worst-informal}
For odd prime $p$, prime power $m=p^k$ and a primitive $m$-th root of unity $\zeta_m$, it is $\Mod_p\P$-hard to compute $\zeta_m$-permanent of $\{0,1\}$-matrices.
More explicitly, for every oracle $\O$ that on input a binary $n\times n$ matrix $X$ outputs the coefficients of $\Per_{\zeta_m}(X)$, it holds that 
\begin{align}\label{eq:intro-worst}
    \PH\subseteq \BPP^{\Mod_p\P} \subseteq \BPP^\O.
\end{align}
\end{theorem}
The first inequality in \eq{intro-worst} follows from the generalized Toda's theorem \cite{toda1992counting,toda1991pp}.
To prove the second, we rely on the fact that for prime power $m=p^k$, the cyclotomic polynomial evaluated on $1$ is $p$.
This means that the summation of the coefficients is congruent to $\Per(X)$ modulo $p$.
Then the second inequality follows from the fact that computing $\Per(X)$ modulo an odd prime is known to be $\Mod_p\P$-hard by Valiant \cite{valiant1979completeness}. 

While our definition of exact computation seems to be natural, one may wonder whether the requirement to output all the coefficients correctly would be too strong to allow the existence of efficient classical algorithms.
We answer this question by studying hardness of approximating $\Per_{\zeta_m}(X)$ in $\ell_2$-norm.
In particular, we show that finding a close enough complex number to $\Per_{\zeta_m}(X)$ already suffices to extract all the coefficients using an $\NP$ oracle.

\subsubsection{Hardness of Approximation}

Next, we show that an oracle that approximates $\Per_{\zeta_m}(X)$ to within polynomial multiplicative error is also $\Mod_p\P$-hard.
\begin{theorem}[Hardness of approximation, informal]\label{thm:approximation-informal}
For odd prime $p$, prime power $m=p^k$ and a primitive $m$-th root of unity $\zeta_m$, it is $\Mod_p\P$-hard to approximate $\zeta_m$-permanent of $\{0,1\}$-matrices.
More explicitly, for every algorithm $\O$ that on input an $n\times n$ binary matrix $X$, outputs an approximation to within multiplicative error $g=\poly(n)$, it holds that 
\begin{align}
    \PH \subseteq \BPP^{\Mod_p\P} \subseteq \BPP^{\NP^\O}.
\end{align}
\end{theorem}
\thm{approximation-informal} implies that an efficient simulation of $\O$ leads to a collapse of the polynomial hierarchy.
The first inequality, again, follows from generalized Toda's theorem \cite{toda1992counting,toda1991pp}.
We then prove the second inequality in the following two steps.
First, we show that there is a deterministic algorithm using polynomially many queries to $\O$ to decrease the error to $\delta=2^{-\poly(n)}$ in $\ell_2$-norm. 
Our algorithm always output a complex number $\delta$-close to $\Per_{\zeta_m}(X)$ successfully.
However, this step alone does not yield a tuple of coefficients that satisfies our definition of exact computation.
We overcome this issue by invoking a variant of Liouville's approximation theorem. 
From the theorem, we show that two arbitrary algebraic integers, when represented as the minimal polynomial in $\zeta_m$ with coefficients in $[-A,B]$ for $A,B=2^{\poly(n)}$ must be $2^{-\poly(n)}$-far in $\ell_2$-distance.
Using an oracle that determines whether there is a $\delta$-close algebraic integer, we can perform a binary search to identify all the coefficients.
Finally, we show determining whether there is an algebraic integer that is close enough to our approximation, and whose coefficients are bounded in $[-A,B]$ that can be solved using a non-deterministic Turing machine.

As a bonus, in \thm{gen-q}, we give evidence that there is no polynomial algorithm for computing most of the roots of unity.

\subsubsection{Average-Case Hardness of Exactly Computing $z$-Permanent}

For the average-case hardness of 1-permanent, Aaronson and Arkhipov \cite{aaronson2011computational} showed that under reasonable assumptions, computing $\Per(X)$ exactly for most of the Gaussian matrices is $\sharpP$-hard. 
Haferkamp, Hangleiter, Eisert, and Gluza \cite{haferkamp2020contracting} proved the hardness for most of the truncated uniform distribution.

In \sec{wa}, we adapt their techniques to show that for every $z$, if computing $z$-permanent is hard in the worst case, via random self-reducibility, computing $z$-permanent of Gaussian matrices is hard.  
This proof utilizes the Berlekamp-Welch algorithm and the fact that $z$-permanent is a low-degree polynomial of its matrix elements.

Moreover, we extend the hardness results to a family of probability distributions that satisfy a property, called the \emph{strong autocorrelation property}.
Specifically, a random matrix is said to satisfy the strong autocorrelation property if the total variational distance between shifted or scaled original distribution and the original one is small (see \dfn{autocor} for a detailed definition).
We show that computing $z$-permanent of matrices drawn from distributions satisfying this property remains to be $\sharpP$-hard. 
The hardness result can thus be applied to, for example, a uniform distribution centered around zero and truncated at some chosen threshold.

\subsubsection{Average-Case Hardness of Approximating $z$-Permanent}

In \sec{van}, 
we show that if one can compute particular $\Omega(n^2)$ points of $\Per_{z}(X)$ in time $\poly(n,1/\epsilon,1/\delta)$, then one can determine $\Per(X)$ for $X$ sampled from a distribution $\H$ that satisfies the strong autocorrelation property. 
This implies that given a distribution $\H$ and $X\sim\H$, approximating $\Per(X)$ is hard. 
Our results work for multiplicative and additive errors.

In addition, we also show that approximating $\Per_z(X)$ is as hard as approximating $\Per_{z^*}(X)$ for any distribution $X$ such that $X,X^*$ are identically distributed. The proof works for $|\Per_z(X)|^2$ as well.

\subsection{Related Work}
As we described in the introduction, permanent and determinant have completely different computational complexity. Any path that connects both cases is interesting from the point of computational complexity. More generally, one can consider such 
\begin{equation}\label{eq:gen}
   \sum_{\nu \in S_n}f(\nu) \prod_{i=1}^{n}A_{i, \nu(i)}
\end{equation} function and ask the complexity of functions in the form of \eq{gen}. 
People have already investigated a few function families including immanants \cite{mertens2011complexity,hartmann1985complexity,burgisser2000computational,de2013determinant,spivak2020immanants,burgisser2000computational} and fermionants \cite{mertens2011complexity,bjorklund2019generalized,rugy2013determinant,de2013determinant}. However, our work is completely different from them in the following perspective:
\begin{itemize}
    \item Unlike immanants and fermionants, $z$-permanent is not a class function. To our best knowledge, the complexity of such matrix function \eq{gen} has only been investigated when $f(\nu)$ is a class function. So one can not use the standard immanant results.
    \item Since $z$-permanent is not a class function, we provide completely new techniques involving algebraic number theory.   
    \item To our best knowledge, All results of approximate hardness in the worst case and average case are investigated for 1-permanent. 
    For $f(\nu)\neq 1$, no approximate hardness results are given previously in the literature. 
    In \sec{app1}, \sec{wa}, and \sec{van}, we provide technical tools which lead to the results of approximate hardness for the general function $f$. 
\end{itemize}

\subsection{Organization}

The rest of the paper is structured as follows. A preliminary background is presented in \sec{pre}. 
In \sec{classical}, for worst-case hardness, we show that for primitive $m$-th root of unity for odd prime power $m=p^k$, exactly computing $z$-permanent is $\Mod_p\P$-hard.  
In \sec{app1}, we show that efficiently approximating the $z$-permanent of binary matrices implies a collapse of the polynomial hierarchy.
In~\sec{wa}, based on the worst-case hardness, we give average-case hardness of exactly computing $z$-permanent for a family of distributions satisfying the strong autocorrelation property.
In~\sec{van}, we move to the average-case hardness. 
Finally, in~\sec{final}, we summarize our results and discuss a few open questions related to $z$-permanent.

\section{Preliminaries}\label{sec:pre}

\subsection{Introduction to Algebraic Number Theory}\label{sec:field}
Here we introduce some facts about field extension \cite{benkart1987abstract} and algebraic number theory \cite{lang2013algebraic, neukirch2013algebraic, weiss1998algebraic} relevant to our work, although in this paper, we will only use facts about cyclotomic fields.

First, we recall basic definitions in abstract algebra and algebraic number theory.
An abelian group is a set of elements closed under its group operation often denoted $+$ (for an additive group) or $\cdot$ (for a multiplicative group) and every element has an inverse.
A commutative ring is an additive group that also has an associative multiplication operation and a multiplicative identity $1$.
A field is a ring, and its non-zero elements form an abelian multiplicative group.
For example, the set $\mathbb Q$ of rational numbers is a field and the set $\mathbb Z$ of integers forms a ring.

A field $K$ is said to be an extension field (or simply extension), denoted $K/F$, of field $F$ if $K$ contains $F$ as a subfield.
The extension field $K$ can be viewed as a vector space over $F$.
The degree of the field extension $K/F$, denoted $[K:F]$, is the dimension of $K$ as a vector space of $F$.
Given a field $F$ and some value $\alpha$, let $F(\alpha)$, called $F$ adjoined $\alpha$, be the smallest field containing $F$ and $\alpha$.
The operation of adjoining an element to a field is a field extension.
For example, $K=\mathbb Q(\zeta_3)=\{a+b\zeta_3:a,b\in\mathbb Q\}$ is an extension field of $\mathbb Q$ of degree two for $\zeta_3=e^{i2\pi/3}$ since every element of $K$ can be written as a vector of two components in $\mathbb Q$.
We also denote $\mathbb Z[x]$ for indeterminant $x$ the ring of polynomials with integer coefficients.
For complex number $\alpha$, the set $\mathbb Z[\alpha]=\{f(\alpha):f(x)\in\mathbb Z[x]\}$ forms a ring. 

A number field is an extension field of $\mathbb Q$ of finite degree.
We recall the definition of number fields and algebraic numbers. 
\begin{definition}
An algebraic number field (or simply number field) is an extension field $K$ of the field of rational numbers $\mathbb {Q}$ such that the field extension $K/\mathbb {Q}$ has finite degree (hence is an algebraic field extension).
\end{definition}

\begin{definition}
A complex number $x\in\mathbb{C}$ is an algebraic number if there exists a nonzero polynomial $f$ with rational coefficients (equivalently, integral) such that $x$ is a root of $f$. 
The set of algebraic numbers $\bar{\mathbb{Q}}$ forms a subfield of $\mathbb{C}.$
\end{definition}
\begin{definition}
A complex number $x\in\mathbb{C}$ is an algebraic integer if there exists a nonzero monic polynomial $f$ (whose leading coefficient equals one) with integral coefficients such that $x$ is a root of $f$. 
\end{definition}
The set of algebraic integer forms a subring of $\bar{\mathbb{Q}}$.
We also have $\alpha \in K$ is an algebraic integer if the minimal monic polynomial of $\alpha$ over $\mathbb {Q}$ is in $\mathbb{Z}[x]$ \cite{lang2013algebraic}. We have a very simple fact about rational field and algebraic number using Gauss's lemma. 
\begin{theorem}[{\cite{lang2013algebraic}}]\label{thm:az}
Let $\O$ be the set of algebraic integers.
The intersection of algebraic integers and $\mathbb{Q}$ is $\mathbb{Z}$, i.e., $\mathbb{Q}\cap\O=\mathbb{Z}.$
\end{theorem}

\begin{definition}[Ring of integers]

The ring of integers $\O_K$ of an algebraic number field $K$ is the ring of all algebraic integers contained in $K$, i.e., $\O_K=\O\cap K$. 

\end{definition}

The cyclotomic field is a number field obtained by adjoining a primitive root of unity to $\mathbb Q$.
In this paper, we denote $\zeta_n$ an $n$th primitive root of unity for integer $n>1$.
Note that a choice of the primitive roots of unity is not unique.
Indeed, for $n>1$, $e^{2\pi ia/n}$ is a primitive root of unity if $\gcd(a,n)=1$.
\begin{definition}[Cyclotomic fields]
The $n$th cyclotomic field is the extension $\mathbb{Q}(\zeta_n)$ of $\mathbb{Q}$ generated by $\zeta_n$. More explicitly:
\begin{equation}
    \mathbb{Q}(\zeta_n)=\left\{ \left.\sum_{i=0}^{n-1} a_i \zeta_n^i \right| a_i \in \mathbb{Q}\right\}.
\end{equation}
\end{definition}
While there can be multiple choices of $\zeta_n$, the field $\mathbb Q(\zeta_n)$ is unique.
This follows from a simple observation: 
Let $\zeta,\xi$ be two distinct $n$th primitive roots of unity. 
Since there exists $a\in\{0,1,\ldots,n-1\}$ and $\gcd(a,n)=1$ such that $\xi=\zeta^a$,
$\mathbb Q(\xi)\subseteq \mathbb Q(\zeta)$ and thus they are equal.
Similarly, we consider the ring 
\begin{equation}
    \mathbb{Z}[\zeta_n]=
    \left\{ \left.\sum_{i=0}^{n-1} a_i \zeta_n^i \right|a_i \in \mathbb{Z}\right\}.
\end{equation}
For cyclotomic field $\mathbb Q(\zeta_n)$, the ring of integers is $\mathbb Z[\zeta_n]$.
\begin{theorem}[{\cite{lang2013algebraic}}]\label{thm:alg}
For integer $n>1$ and $K=\mathbb Q(\zeta_n)$, $\O_K=K\cap \O=\mathbb{Z}[\zeta_n]$.
\end{theorem}
We will need some useful properties about Euler's totient function \cite{lang2013algebraic}.
\begin{definition}[Euler's totient function {\cite{lang2013algebraic}}]\label{dfn:euler}
Euler's totient function $\phi$ counts the number of positive integers up to a given integer $n$ that are relatively prime to $n$. More explicitly, it can be computed through the following formula:
\begin{equation}
    {\displaystyle \phi (n)=n\prod _{p\mid n}\left(1-{\frac {1}{p}}\right),}
\end{equation}
where the product is over the distinct prime numbers dividing $n.$
\end{definition}

We will need the following lower bound of $\phi(n)$ \cite{rosser1962approximate}.
\begin{theorem}[{\cite{rosser1962approximate}}]
For $n>2$,
\begin{equation} 
  \phi (n)>{\frac {n}{e^{\gamma }\;\log \log n+{\frac {3}{\log \log n}}}}.
\end{equation}
simply speaking: the order of $\phi(n)$ is nearly $n$. So this implies that we always have
$\phi(\Omega(n^c))>n$ for any $c>1$ for large enough $n.$ Here $\gamma$ is Euler–Mascheroni constant.
\end{theorem}

For an $n$th primitive root of unity $\zeta_n$, the $n$th cyclotomic polynomial $\Phi_n$ is the minimal polynomial of $\zeta_n$ over $\mathbb Q$.

\begin{definition}[{\cite{lang2012cyclotomic}}]\label{dfn:cyclo}
For integer $n\geq 1$, the $n$th cyclotomic polynomial  is defined as
\begin{equation}
{\displaystyle \Phi _{n}(x)=\prod _{\stackrel {1\leq k\leq n}{\gcd(k,n)=1}}\left(x-e^{2i\pi {\frac {k}{n}}}\right)}.
\end{equation}
The polynomial $\Phi_n(x)$ has integer coefficients, i.e., $\Phi_n(x)\in\mathbb Z[x]$.
\end{definition}
We include some basic properties of cyclotomic polynomials useful in our work.
Since $\Phi_{n}(x)$ is irreducible over $\mathbb{Q}$, so it is the minimal polynomial of $\zeta_n$ over $\mathbb{Q}$,  the degree $[\mathbb{Q}(\zeta_n) : \mathbb{Q}]=\text{deg }\Phi_{n}= \phi(n)$. The minimal polynomial of every element in $\mathbb{Q}(\zeta_n)$ over $\mathbb Q$ has 
degree at most $\phi(n)$. The equivalent statement is that the following set
$\{ e^{2 \pi a i j/q} : 0 \leq j < \phi(q) \}$
is linearly independent over $\mathbb{Q}$.
For $n\geq 2$, 
\begin{align}
\Phi _{n}(0) &= 1,\nonumber \\
\Phi_{n}(1) &= 1, \quad\text{ if $n$ is not a prime power}, \nonumber\\
\Phi _{n}(1) &= p, \quad\text{  if $n=p^{k}$ is a prime power with $k \geq 1$.}
\end{align}

\begin{theorem}[Bounds for the coefficients of cyclotomic polynomials \cite{vaughan1975bounds}]\label{thm:cyclotomic-coefficient-bound}
Let $\Phi_n(x):=\sum_{i=0}^{\phi(n)}a_ix^i$ be the $n$th cyclotomic polynomial.
The maximum absolute value of coefficient $a_i$ of $\Phi_n(x)$ satisfies
\begin{equation}
    \max_{i\in\mathbb N, 0\leq i\leq \phi(n)} |a_i| \leq e^{\frac{1}{2}d(n)\log n},
\end{equation}
where $d(n)$ is the divisor function which counts the number of divisors of $n$ (including $1$ and $n$) for integer $n\geq 1$.\footnote{For instance, the divisors of $6$ are $1$, $2$, $3$, and $6$, and thus $d(6)=4$.} 
\end{theorem}
For $n=p_1^{a_1}\ldots, p_k^{a_k}$, $d(n)= (a_1+1)\ldots (a_k+1)$.
Since for prime $p\geq 2$ and integer $a\geq 1$, $a+1\leq p^{a}$, we can get the trivial bound $d(n)\leq n$.
By \thm{cyclotomic-coefficient-bound}, we can represent $\Phi_n(x)$ using at most only $O(n\log n)$ bits.

The following theorem will be useful.
\begin{theorem}[Liouville's approximation theorem {\cite{levesque2012approximation}}]\label{thm:liouville}
Let $\alpha \in \mathbb{C}$ be an algebraic number with minimal polynomial $f(x)=a_0 x^d+a_1x^{d-1}+..+a_d$ of degree $d$. Then for any $a,b \in \mathbb{Z}$ that are relatively prime,
\begin{equation}
    \left|\alpha-\frac{a}{b}\right|\geq \frac{1}{b^d a_0 (2|\bar{\alpha}|+1)^{d-1}}
\end{equation}
where $f(x)=a_0 \prod_{\sigma}(x-\sigma(\alpha))$, and $|\bar\alpha|$ the maximum complex modulus of the algebraic conjuguates of $\alpha$ in $\mathbb{C}$, i.e., $|\bar{\alpha}| =\max|\sigma(\alpha)|$.
\end{theorem}

In particular, since $0$ and $1$ are relatively prime, applying \thm{liouville}, for every algebraic number $\alpha$ whose minimal polynomial is $f(x)=a_0x^d+\ldots+a_d$, $|\alpha|\geq \frac{1}{a_0(|\bar\alpha|+1)^{d-1}}$. 
For two algebraic integers $\alpha,\beta\in\mathbb Z[\zeta_n]$, $\alpha-\beta\in\mathbb Z[\zeta_n]$, and $|\alpha-\beta|\geq \frac{1}{(|\overline{\alpha-\beta}|+1)^{\phi(n)-1}}$.

\begin{theorem}[Fermat's little theorem~\cite{weisstein2004fermat}]\label{thm:f}
Let $p$ be a prime number. 
Then for any integer $a$, the number $a^p-a$ is an integer multiple of $p$. In the notation of modular arithmetic, this is expressed as
\begin{equation}
    a^p \equiv a \pmod p.
\end{equation}
\end{theorem}
\thm{f} implies
\begin{equation}
    a^{p^k} \equiv a \pmod p.
\end{equation}
for $k \geq 1$ since by \thm{f}, $a^{p^k}\equiv a^{p^{k-1}} \pmod p$.

\begin{theorem}\label{thm:finite-field}
For field $F$ of characteristic $p$ and $a,b\in F$, 
\begin{proof}
By the binomial theorem,
\begin{equation}
    (a+b)^p=a^p+\binom{p}{1}a^{p-1}b+\binom{p}{2}a^{p-2}b^{2}+...+\binom{p}{p-1}ab^{p-1}+b^p.
\end{equation}
We just need the following Lemma: The prime $p$ divides each binomial coefficient
$\binom{p}{r}$
for $1 \leq r \leq p-1.$ The proof of lemma is simple by definition:
\begin{equation}
  \binom{p}{r}=\frac{p \cdot (p-1) \cdot ...\cdot (p-r+1)}{1 \cdot 2 \cdot ...\cdot r}.  
\end{equation}
The result follows since $p$ divides the top but not the bottom.
\end{proof}
\end{theorem}

\subsection{Properties of $z$-permanent}

Let $\mathbb T:=\{z\in\mathbb{C}:|z|=1\}$ be the unit circle in the complex plane.
For matrix $X\in\mathbb C^{n\times n}$, the $z$-permanent of $X$ is defined as 
\begin{align}\label{def:z}
  \Per_z(X) := \sum_{\sigma\in S_n} z^{\ell(\sigma)}  \prod_{i} X_{i,\sigma(i)},
\end{align}
where $S_n$ is the symmetric group on $n$ elements and $\ell(\sigma)$ is the inversion number of $\sigma$.
The inversion number is defined as the number of pairs $(i,j)$ such that $i<j$ and $\sigma(i)>\sigma(j)$.
It is clear that $\ell(\sigma)\leq \binom{n}{2}=\frac{1}{2}n(n-1)$.
This implies that we can write a $z$-permanent as a polynomial in $z$:
\begin{align}\label{eq:per-polynomial}
  \Per_z(X) = \sum_{\ell=0}^{\binom{n}{2}} z^\ell \left(\sum_{\sigma:\ell(\sigma)=\ell} \prod_{i=1}^n X_{i,\sigma(i)}\right).
\end{align}
Note that by \eq{per-polynomial}, $z$-permanent is a polynomial with degree ${n \choose 2}$ in terms of $z.$ 

Here we list basic properties of $z$-permanent \cite{yang1991q, citation-0} for every matrix $X\in\mathbb C^{n\times n}$. 
Most of them are exact the same as permanent and determinant.
\begin{enumerate}
\item $\Per_{z}(X)$ is a multilinear function of the rows and columns.

\item For block triangular $X$, $\Per_z(X)$ is the product of the $z$-permanent of its diagonal blocks.

\item $\Per_{z}(X)=\Per_{z}(X^{T})$ where $X^T$ is the transpose of $X$.

\item (Expansion Theorem) Let $X^{ij}$ denote the $(i, j)$ minor of $X$ and $X_{ij}$ denote the element of $X.$ Then
\begin{equation}\label{eq:expan}
   \Per_z(X) = \sum_{k=1}^{n}X_{1k}z^{k-1}\Per_z(X^{1k})
    = \sum_{k=1}^{n}X_{jk}z^{k-1}\Per_z(X^{jk})
\end{equation}
for $j=1,2,\ldots,n$.

\item For integer matrix $X$ and $n$-th root of unity $z$, $\Per_z(X) \in \mathbb{Z}[\zeta_n]$.
Thus, the minimal polynomial of $\Per_z(X)$ over $\mathbb{Q}$ has leading coefficient $1$ and its degree is at most $\phi(n)$.

\end{enumerate}

We give the explicit formula for the simplest case $n=2$ and $n=3$:
For $n=2$,
\begin{align*}
    \Per_z(X)=(X_{11}X_{22}+X_{21}X_{12}z).
\end{align*}
For $n=3$,
\begin{align*}
    \Per_z(X)=X_{13} X_{22} X_{31}z^3+\left(X_{12} X_{23} X_{31}+X_{13} X_{21} X_{32}\right)z^2+
    \\ \nonumber
    \left(X_{11} X_{23} X_{32}+X_{12} X_{21} X_{33}\right)z+X_{11} X_{22} X_{33}.
\end{align*}

\subsection{Common Tasks Related to Permanent}
We recall the following task defined by Aaronson and Arkhipov \cite{aaronson2011computational}. We will discuss some of their generalizations in our paper. 
\begin{problem}[$\GPE_{\times}$]
  Given as input $X\sim\N_{\mathbb{C}}(0,1)^{n\times n}$ of i.i.d Gaussians (Complex normal distribution) and $\epsilon,\delta>0$, estimate $\Per(X)$ to within error $\pm\epsilon|\Per(X)|$ to probability at least $1-\delta$ over $X$.
\end{problem}

\begin{conjecture}\label{conj:GPE}
  $\GPE_\times$ is $\sharpP$-hard.
  In other words, if $\O$ is any oracle that solves $\GPE_\times$, then $\P^\sharpP\subseteq\BPP^{\O}$.
\end{conjecture}

\begin{conjecture}\label{conj:pacc}
  There exists a polynomial $p$ such that for positive integer $n$, real number $\delta>0$, 
  \begin{align}
    \Pr_{X}\left[ |\Per(X)|^2 \geq \frac{\sqrt{n!}}{p(n,1/\delta)}\right] \geq 1-\delta. 
  \end{align}
\end{conjecture}

\section{Hardness of $z$-permanent for $\{0,1\}$-matrices}\label{sec:classical}

In this section, we investigate the hardness of $z$-permanent.
We only consider the case that $z$ is root of a unity, i.e., there is positive integer $m$ such that $z^m=1$, for the following two reasons. 
First, the set $\{z=e^{2 \pi i a/b}:a,b\in\mathbb{N}\}$ is dense on the unit circle, so for every $z\in\mathbb T$, there exists an arbitrarily close element in this set.
Second, from the computational point of view, as we will see, it is more relevant and convenient to give a definition of exact computation with a rational multiple of $2\pi$. 
Thus, in this paper, we always assume that $z$ is root of a unity, and write $z=\zeta_{m}$ for primitive $m$-th root for unity for integer $m\geq 2$.

We start with the cubic root of unity, showing that computing $z$-permanant for $z$ satisfying $z^3=1$ is $\Mod_3\P$-hard.
Thus the existence of an algorithm for computing the $z$-permanent implies the collapse of the polynomial hierarchy by Toda's theorem.

\subsection{Cube Roots of Unity}

Recall that from algebraic number theory, for $z=\zeta_p$ and binary matrix $X\in\bit^{n\times n}$, $\Per_z(X)$ can be written as a polynomial in $z$:
\begin{align}\label{eq:per-poly}
  \Per_z(X) = \sum_k A_k z^k,
\end{align}
where $A_k:=\sum_{\sigma:\ell(\sigma)\equiv k\pmod p} \prod_{i=1}^n X_{i,\sigma(i)}$.
Furthermore, for $z=\zeta_p$, $\Per_z(X)\in\mathbb Z[\zeta_p]$, and thus is an algebraic integer in the number field $\mathbb Q[\zeta_p]$.
We will link the observation to hardness of computing $1$-permanent and the complexity $\Mod_p\P$, defined as follows.

\begin{definition}[$\Mod_k\P$]
For any integer $k \geq 2$, let $\Mod_{k}\P$ be the class of decision problems solvable by a polynomial time non-deterministic machine which rejects if the number of accepting paths is divisible by $k$, and accepts otherwise. When $k = 2$, $\Mod_{k}\P$ is also known as parity $\P$, and
denoted $\oplus \P$.
\end{definition} 

It is well known due to Valiant that computing $\Per(X)\bmod p$ is $\Mod_p\P$-complete for prime $p>2$ \cite{valiant1979completeness}.
Interestingly, we can reduce the problem of computing $\Per(X)\bmod p$ to the problem of computing $\Per_{\zeta_p}(X)$ in the special case $p=3$.

\begin{theorem}[Valiant \cite{valiant1979completeness}]\label{thm:vali}
The problem of computing $\Per(M) \bmod p$ for a square matrix $M \in \mathbb{F}_p^{n \times n}$
is $\Mod_{p}\P$-complete for any prime $p>2$.
\end{theorem}

Now we state the theorem for this section.

\begin{theorem}\label{120}
For every cube roots of unity $\zeta$, 
computing $\Per_{\zeta}(X)$ and $|\Per_{\zeta}(X)|^2$ are both $\Mod_{3}\P$-hard.
\end{theorem}
\begin{proof}
For $\zeta=1$, the theorem holds due to Valiant \cite{valiant1979completeness}, and thus it suffices to prove the theorem for $\zeta\in\{e^{i2\pi/3},e^{i4\pi/3}\}$.
For $\zeta=e^{i2\pi/3}$, we can write $\Per_{\zeta}(X)$ as
\begin{align}\nonumber
  \Per_\zeta(X)
  &=A_0+\zeta A_1+\zeta^2 A_2 \\\nonumber
  &=A_0+\left(-\frac{1}{2}+\frac{\sqrt{3}i}{2}\right)A_1 +\left(-\frac{1}{2}-\frac{\sqrt{3}i}{2}\right)A_2 \\
  &=\frac{2A_0-A_1-A_2}{2}+i\sqrt{3}\frac{A_1-A_2}{2}
\end{align}
If one can compute $\Per_{\zeta}(X)$, then one can compute $\Per(X)\bmod 3$ by taking the real part of $\Per_\zeta(X)$ since 
\begin{align}
  2A_0-A_1-A_2 
  \equiv 2\Per(X) \pmod 3,
  \end{align}
  and $2$ is a unit in the ring $\mathbb Z_3$.
  The same reasoning holds for $\zeta=e^{i4\pi/3}$ by interchanging $i$ and $-i$.

It was shown in Valiant's original paper \cite{valiant1979complexity} that the permanent modulo $m=2^k$ can be computed in time $n^{O(k)}$, but for all $m\neq 2^k$, it is $\NP$-hard under randomized reductions \cite{curticapean2015parameterizing, datta2021parallel}.

For $|\Per_\zeta(X)|^2$, we observe that 
\begin{align}\nonumber
  |\Per_{\zeta}(X)|^2 
  &\equiv |\Per(X)|^2 \pmod 3. 
\end{align}
Thus we can obtain $|\Per(X)|^2\bmod 3$ if $|\Per_\zeta(X)|^2$ can be exactly computed. 
\end{proof}

The generalized Toda's theorem \cite{toda1992counting} states that a few counting classes are as hard as the polynomial hierarchy.

\begin{theorem}[Toda theorem~\cite{toda1991pp,toda1992counting}]\label{thm:toda}
For prime $k>2$, let $A$ be one of the counting classes $\Mod_{k}\mathsf{P}, \mathsf{\#P}, \mathsf{PP},$ or $\mathsf{GapP}$. Then $\mathsf{PH} \subseteq \mathsf{BPP}^{A}$.
\end{theorem}
By Toda's theorem, we immediately have the following theorem.
\begin{theorem}\label{thm:prime3}
  For $\zeta\in\{1, e^{i2\pi/3}, e^{i4\pi/3}\}$, let $\O$ be any oracle oracle that on input binary matrix $X\in\bit^{n\times n}$, outputs $\Per_\zeta(X)$. Then
\begin{equation}
  \mathsf{PH} \subseteq \BPP^{\Mod_3\P} \subseteq
  \BPP^{\O}.
\end{equation}
In other words, if there is an efficient classical simulation of $\O$, then the polynomial hierarchy collapses to the second level.
\end{theorem}
If the polynomial hierarchy is infinite, \thm{prime3} excludes the existence of an efficient classical algorithm for exactly computing $\zeta$-permanent.

\subsection{Primitive $m$th roots of unity for prime power $m$}\label{sec:prime-power}

Next, we turn our attention to any prime number $p=n^{O(1)}$.
Since in general $\Per_{\zeta_p}(X)$ can be an irrational number in the complex plane, any truncation in finite precision only results in an approximation.
In this section, we focus the hardness of computing $\Per_{\zeta_p}(X)$ exactly, and thus the first question we must clarify is how we define ``computing $\Per_{\zeta_p}(X)$ exactly'' for binary $X$.

For binary matrix $X$, $A_k\in\mathbb Z$ (see \eq{per-poly}).
Thus $\Per_{\zeta_p}(X)$ can be written as a polynomial in $\zeta_p$ of degree at most $p-1$ with integer coefficients.
We can define an exact representation of $\Per_{\zeta_p}(X)$ as a set of coefficients.
\begin{definition}[Representation of an algebraic integer]\label{dfn:exact-representation}
  A representation of algebraic integer $\alpha\in\mathbb Z[\zeta_p]$ is a tuple of integers $(b_0,b_1,\ldots,b_{p-1})$ satisfying
  \begin{align}\label{eq:alpha-poly}
    \alpha = b_0 + b_1\zeta_p + \ldots b_{p-1}\zeta_p^{p-1}.
  \end{align}
  Note that the representation is not unique since $\zeta_p$ is a root of $f(x)=1+x+\ldots+x^{p-1}$, and thus $\alpha+cf(\zeta_p)=\alpha$ for any integer $c\in\mathbb Z$. 
\end{definition}

In general, one can define a representation of algebraic integers  $\alpha \in \mathbb{Z}[\zeta_n]$ for general $n$ in \dfn{general-n}. However, in this section, the \dfn{exact-representation} is sufficient for showing the hardness. 
Notice that one can also define such representation by a tuple of integers $(b_0,b_1,\ldots,b_{p-2})$ satisfying $\alpha = b_0 + b_1\zeta_p + \ldots b_{p-2}\zeta_p^{p-2}$ since one has $\zeta_{p}^{p-1}=-1-\zeta_p-\zeta_p^2-...-\zeta_{p}^{p-2}.$

An algorithm which computes $\Per_{\zeta_p}(X)$ exactly outputs a tuple of coefficients $(b_0,\ldots,b_{p-1})\in\mathbb Z^{p}$ such that \eq{alpha-poly} holds.
This definition is without loss of generality for our purpose. 
In \sec{stod}, we will construct an algorithm that maps an approximation in the complex plane to an exact representation using an $\NP$ oracle.

\begin{theorem}\label{thm:mod-p}
  For prime $p$, let $\zeta_p$ be a primitive $p$-th root of unity and $\O$ be any oracle that on input a binary matrix $X\in\bit^{n\times n}$, outputs a representation of $\Per_{\zeta_p}(X)$.
  Then $\BPP^{\Mod_p\P}\subseteq \BPP^{\O}$.
\end{theorem}
\begin{proof}
For a binary matrix $X$, let the output of $\O(X)$ be $(a_0,\ldots,a_{p-1})$.
Also, let $A_i$ be defined as in \eq{per-poly}.
By the correctness of $\O$, 
\begin{equation}\label{eq1}
    \Per_{\zeta_p}(X)=\sum_{i=0}^{p-1}A_i \zeta_p^i=\sum_{i=0}^{p-1}a_i \zeta_p^i.
\end{equation}

Note that it does not necessarily hold that for every $i$, $A_i=a_i$, but the polynomials are identical congruent to the minimal polynomial of $\zeta_p$.
In particular, for prime $p$, since $0=\sum_i (A_i-a_i)\zeta_p^i$, 
\begin{equation}
    \left(\sum_{i=0}^{p-1}(A_i-a_i)\zeta_p^i\right)^p=\sum_{i=0}^{p-1} (A_i-a_i)^p+ph=0
\end{equation}
where $h \in \mathbb{Z}[\omega]$.
Since the left side of the equality
\begin{align}
  \sum_{i=0}^{p-1} (A_i-a_i)^p = -ph
\end{align}
is an integer, $h\in\mathbb Z[\zeta_p]\cap \mathbb Q\subseteq\mathbb Z$.
This means that $p$ divides $\sum_i (A_i-a_i)^p$, and
\begin{align}\nonumber
    0 
  &\equiv \sum_{i=0}^{p-1} (A_i-a_i)^p  \\\nonumber
    &\equiv \sum_{i=0}^{p-1} A_i^p-\sum_{i=0}^{p-1} a_i^p \\
    &\equiv \sum_{i=0}^{p-1} A_i -a_i  \pmod p.
\end{align}
The last equality holds from Fermat's little theorem (\thm{f}). 

In summary, given access to an oracle $\O$ which exactly computes $\Per_{\zeta_p}(X)$ for binary $X\in\bit^{n\times n}$, the following algorithm computes $\Per(X)\bmod p$: 
\begin{itemize}
\item Run $(a_0,\ldots,a_{p-1})\gets\O(X)$.
\item Output $\sum_{i=0}^{p-1}a_i \bmod p$.
\end{itemize}
\end{proof}

By Toda's theorem \thm{toda}, we have the following corollary.
\begin{corollary}\label{prime}
For prime $p>3$, let $\O$ denote any oracle which computes $\Per_{\zeta_p}(X)$ exactly for binary matrix $X\in\bit^{n\times n}$. Then 
\begin{equation}
  \PH \subseteq \BPP^{\Mod_p\P}\subseteq  \BPP^{\O}. 
\end{equation}
\end{corollary}
We can easily extend our result to a prime power $q=p^k$ for prime $p$.
\begin{theorem}\label{thm:mod-p}
  For prime $p$ and prime power $q=p^k$, let $\zeta_q$ be a primitive $q$-th root of unity and $\O$ be any oracle that on input a binary matrix $X\in\bit^{n\times n}$, outputs an exact representation of $\Per_{\zeta_q}(X)$.
  Then $\BPP^{\Mod_p\P}\subseteq \BPP^{\O}$.
\end{theorem}
\begin{proof}
The idea basically follows from the proof of \thm{mod-p}.
For $\zeta_{q}=e^{2 \pi i a/q}$ and $\gcd(a,p)=1$, 
we have
\begin{equation}
    \left(\sum_{i=0}^{q-1}(A_i-a_i)\zeta_{q}^i \right)^{q}=\sum_{i=0}^{q-1} (A_i-a_i)^{q}+ph=0
\end{equation}
where $h \in \mathbb{Z}[\omega]\cap\mathbb Q\subseteq \mathbb Z$.
We then use Fermat's little theorem (\thm{f}) which imples that $a^{q}\equiv a(\bmod p)$, and
\begin{equation}
    \sum_{i=0}^{q-1} (A_i-a_i)+ph=0.
\end{equation}
So then we have (using finite field property \thm{finite-field})
\begin{equation}
    \Per(X) \equiv \sum_{i=0}^{q-1}A_i\equiv\sum_{i=0}^{q-1}a_i \pmod p,
\end{equation}
which finish the proof.
\end{proof}

We end this section by summarizing the classical hardness of $\Per_{z}(X)$ for such particular $z$.
\begin{theorem}\label{thm:classical hard}
For odd prime $p$ and prime power $q=p^k$, let $\zeta_q$ be a primitive $q$-th root of unity.
Exactly computing $\zeta_q$-permanent for binary matrices is not in $\PH$; otherwise, the polynomial hierarchy collapses.
More precisely, for any oracles $\O$ which computes the $\zeta_q$-permanent of binary matrices, the following containments hold:
\begin{align*}
    \PH \subseteq    \BPP^{\Mod_p\P}\subseteq \BPP^{\O}.
\end{align*}
\end{theorem}

Unfortunately, our method does not generalize to a primitive $m$-th root of unity for $m$ divisible by more than one prime. 
While our result does not rule out the existence of an efficient solver for non-prime powers, it is interesting to see whether there exists a polynomial-time algorithm for any of these points. Since it is not so relevant to our result, we leave the complexity of computing general composite numbers as future work. 

For $p=2$, we notice that our method with oracle which computes $\zeta_{2^k}$-permanent implies one can compute $\Per(X) \pmod 2$. 
However, since $\Per(X) \equiv\det(X) (\bmod 2),$ we do not have the hardness results for computing $\zeta_{2^k}$-permanet Indeed, it is known that permanent modulo $2^k$ admits an efficient algorithm \cite{datta2021parallel,curticapean2015parameterizing}.
Any hardness results for $\zeta_{2^k}$-permanent cannot be implied by the hardness of permanent modulo $2^k$.

Unfortunately, our method cannot be used to establish hardness for rational numbers with a non-prime-power denominator. 
\subsection{Hardness of Rational Phases}
\paragraph{}

Recall that for every $z\in\mathbb T$, the $z$-permanent of a binary matrix is always an integral polynomial in $z$ of degree at most $\binom{n}{2}$.
For a primitive $m$-th root of unity $z=\zeta_m$, we can further reduce the degree by taking the congruence to the minimal polynomial of the degree of $\zeta_m$, when $\phi(m)<\binom{n}{2}$.
In the previous sections, we have applied these facts to reduce the problem of computing $1$-permanent modulo $m$ to $\zeta_m$-permanent for prime power $m$.

While we do not have a general result for every non-prime-power rational phases, in this section, we consider a special case where $\phi(m)\geq\binom{n}{2}$.
In this case, since the monomials $1,\zeta_m,\zeta_m^2,\ldots,\zeta_m^{\binom{n}{2}}$ are linearly indepenent, $1$-permanent can be directly obtained from $\zeta_m$-permanent.

\begin{theorem}\label{thm:gen-q}
Let $\O$ denote any oracle which input are binary matrix $X$ and $m:\mathbb{N}\to\mathbb N$ such that $\phi(m(n))\geq\binom{n}{2}+1$, it outputs $\Per_{\zeta_m}(X)$ exactly. Then $\P^{\# \P} \subseteq \P^{\mathcal{O}}$. 
\end{theorem}
\begin{proof}
Let $d:=\binom{n}{2}$.
Computing $\O(m,X)$ yields a tuple of coefficients $(a_0,\ldots,a_d)$ satisfying 
\begin{equation}
    \Per_{\zeta_m}(X)=\sum_{i=0}^{d}A_i\zeta_m^{i}=\sum_{i=0}^{d}a_i\zeta_m^{i}.
\end{equation}
where $A_i=\sum_{\sigma \in \mathbb S_N, \ell(\sigma)= i} \left(\prod_{j=1}^n X_{j,\sigma(j)}\right)$.
Since the monimials $1,\zeta_m,\ldots,\zeta_m^{d}$ are linearly dependent for $\phi(m)\geq d+1$, 
$A_i=a_i$ for every $0\leq i\leq d$.
This implies that $\sum_{i=0}^d a_i=\Per_{m}(X)$.
\end{proof}
However, for cases with non-prime-power $m=O(1)$, neither of the methods seems to work.
We leave the hardness of primitive $m$-th roots of unity for fix non-prime-power $m$ as an open question.

\section{Hardness of Approximating $\zeta_m$-permanent}\label{sec:app1}

In \sec{classical}, we have shown that for prime power $m=p^k$ for odd prime $p$, exactly computing $\Per_{\zeta_m}(X)$ is as hard as computing $\Per(X)\bmod p$.
The task is known to be $\Mod_p\P$-hard.
In this section, we show that approximating $\Per_{\zeta_m}(X)$ for integer $m$ and binary matrices is as hard as exactly computing it. 

In \sec{amplification}, we show that approximating $|\Per_{\zeta_m}(X)|^2$ to multiplicative error within $g=\poly(n)$ is as hard as approximating $\Per_{\zeta_m}(X)$ to multiplicative error $2^{-\poly(n)}$, up to a polynomial multiplicative overhead.
In \sec{stod}, we show that with an $\NP$ oracle, exactly computing $\Per_{\zeta_m}(X)$ is as hard as approximating it to error $2^{-\poly(n)}$.
Combining the results in these two sections, we conclude that an efficient approximation algorithm would also imply a collapse of the polynomial hierarchy.

\subsection{Reducing the Error}\label{sec:amplification}

In this section, we show that for a prime number $m$ and an $m$-th primitive root of unity $\zeta_m$, approximating $|\Per_{\zeta_m}(X)|^2$ to multiplicative error $g=n^{O(1)}$ is as hard as approximating $\Per_{\zeta_m}(X)$ to multiplicative error $1+2^{-n^{O(1)}}$.

First we prove the following lemma which will be useful later.
\begin{lemma}\label{lem:error-propagation}
  For integer $n\geq 1$ and $\delta\leq \frac{1}{2n}$, let $\alpha_1,\ldots,\alpha_n$ be complex numbers and $\tilde\alpha_1,\ldots,\tilde\alpha_n$ satisfying for $i\in[n]$, 
  \begin{align}\label{eq:alpha-i}
    |\alpha_i-\tilde\alpha_i| \leq \delta |\alpha_i|.
  \end{align}
  Then it holds that 
  \begin{align}
    |\alpha_1\ldots\alpha_n - \tilde\alpha_1\ldots\tilde\alpha_n|
    \leq 2n\delta |\alpha_1\ldots\alpha_n|.
  \end{align}
\end{lemma}
\begin{proof}
  We prove the lemma by induction.
  For $n=1$, the lemma obviously hold.
  Suppose that the lemma holds for $n=k$.
  Then for $n=k+1$, let $A:=\alpha_1\ldots\alpha_k$ and $\tilde A:=\tilde\alpha_1\ldots\tilde\alpha_k$.
  By the inductive hypothesis, $|\tilde A-A|\leq 2k\delta|A|$.
  Then 
  \begin{align}\nonumber
    \frac{|\alpha_{k+1}A-\tilde\alpha_{k+1}\tilde A|}{|\alpha_{k+1}A|}
    &\leq \frac{|\alpha_{k+1}-\tilde\alpha_{k+1}|}{|\alpha_{k+1}|}
    + \frac{|\tilde\alpha_{k+1}|}{|\alpha_k|}\frac{|A-\tilde A|}{|A|} \\\nonumber
    &\leq \delta + (1+\delta) 2k\delta \\\nonumber
    &\leq (2k+1)\delta + 2k\delta^2 \\
    &\leq (2k+2)\delta.
  \end{align}
  The first inequality holds by \eq{alpha-i} and $|\tilde\alpha_{k+1}|\leq |\alpha_{k+1}-\tilde\alpha_{k+1}|+|\alpha_{k+1}|\leq (1+\delta)|\alpha_{k+1}|$.
  The last inequality holds since $2k\delta^2\leq 2n\delta^2\leq \delta$.
\end{proof}

We can decrease the error of our approximation algorithm from a multiplicative factor $g=n^{O(1)}$ to $1+2^{-n^{O(1)}}$ using a polynomial overhead.

\begin{theorem}\label{thm:we to wa}
  For $g\in[1,\poly(n)]$, prime number $m$, and a primitive $m$-th root of unity $\zeta_m$, if there is an algorithm $\O$ which on input $X\in\bit^{n\times n}$, approximates $|\Per_{\zeta_m}(X)|^2$ to within a multiplicative factor of $g$, 
  then there exists an algorithm which outputs an esatimate $\tilde P$ of $\Per_{\zeta_m}(X)$ satisfying
  \begin{align}
    \frac{|\Per_{\zeta_m}(X)-\tilde P|}{|\Per_{\zeta_m}(X)|} \leq \epsilon
  \end{align}
 using $O(mn^2g^4\log n+ n g^4\log g + ng^4\log(1/\epsilon))$ queries to $\O$.
\end{theorem}

\begin{proof} 
  Let $z=\zeta_m$.
  We describe the algorithm $\A^\O$ as follows.
  First $\A^\O$ decides if $|\Per_z(X)|^2=0$.
  If $|\Per_z(X)|^2\neq 0$, then $\O$ returns a value 
  \begin{align}
  p\in [g^{-1}|\Per_z(X)|^2, g |\Per_z(X)|^2]
  \end{align}
  for $g>1$.
  This implies that $p>0$.
  Otherwise, $\O(X)=0$.
  This implies that there is an algorithm which decides if $\Per_z(X)=0$.
  Furthermore, if $p=0$, $\A$ outputs $0$.

  It suffices to consider the case $|\Per_z(X)|^2\neq 0$.
  Since $X$ is binary-valued, we can write
  \begin{align}
    |\Per_z(X)|^2 = a_0 +a_1 z+\ldots + a_{m-1}z^{m-1},
  \end{align}
  where each $a_i$ is a non-negative integer,
  provided that $z$ is a primitive $m$-th root of unity.
  
  The algorithm $\A^\O$ finds a sequence of matrices $X_0=X,X_1,\ldots,X_{n-1}$ such that each $X_i\in\mathbb C^{(n-i)\times (n-i)}$ and $X_{i+1}$ is a submatrix of $X_i$ obtained by deleting the first column and one of the rows $k$ determined as follows: 
  starting from $k=1$, repeat incrementing $k$ until the submatrix $W$ obtained by deleting the $k$-th row and the first column satisfies that $|\Per_z(W)|^2\neq 0$.
  Finally set $X_{i+1}=W$.
  Since $\Per_z(X)\neq 0$, it is guarantee that there exists $W$ such that $|\Per_z(W)|\neq 0$.
  Computing the sequence of matrices takes $O(n^2)$ queries.

  The algorithm inductively approximates $\frac{\Per_z(X_{k-1})}{\Per_z(X_k)}$ to error $O(\epsilon/n)$ using $O(mng^4\log n+g^4\log g+g^4\log(1/\epsilon))$ queries to $\O$ from $k=n-1$ to $0$ and the telescoping product
  \begin{align}
    \Per_z(X) = 
    \frac{\Per_z(X_0)}{\Per_z(X_1)}\cdot
    \frac{\Per_z(X_1)}{\Per_z(X_2)}
    \cdots \frac{\Per_z(X_{n-2})}{\Per_z(X_{n-1})}\cdot \Per_z(X_{n-1}).
  \end{align}
  Computing the base case $\Per_z(X_{n-1})$ is trivial. 
  Applying \lem{error-propagation}, this implies that an approximation to error $\epsilon$ takes $O(mn^2g^4\log n + g^4n\log g)$ queries.

  In the following proof, we will only consider the case from $Y:=X_1$ to $X_0$, and the same argument applies to $k=2,\ldots,n-1$.
  For simplicity, we only prove the case where $Y$ is obtained from $X$ by deleting the first row and the first column.
  Other cases (i.e., deleting the $i$-th row of $X$ for $i>1$) can be applied using the same argument.

  We consider the quantity $\Per_z(X)-\alpha\Per_z(Y)$.
  Note that for a primitive root of unity $z$, both $\Per_z(X)$ and $\Per_z(Y)$ are algebraic integers.
  This implies that the ratio $\alpha=\frac{\Per_z(X)}{\Per_z(Y)}$ is an algebraic number.
  To see there is a lower bound between two algebraic numbers, we apply a variant of the Liouville's approximation theorem, stated in \thm{liouville}.
  
Since $\Per_z(X)$ and $\Per_z(Y)$ are both algebraic integers, their minimal polynomials are monic.
Furthermore, for $\Per_z(X)=\mu$, $\max|\sigma(\mu)|\leq \Per(X)\leq n!$, and similarly the same upper bound holds for $Y$.
Setting $a=0$ and $b=1$, $|\Per_z(Y)|\geq 1/M$ and $\Per_z(X)\geq 1/M$, where $M:=(2n!+1)^{m-2}$, and
\begin{align}
  \frac{1}{n!M}\leq \frac{|\Per_z(X)|}{|\Per_z(Y)|} \leq n!M.
\end{align}

  The quantity $\Per_z(X)-t\Per_z(Y)=\Per_z(X^{[t]})$, where $X_{ij}^{[t]}=X_{ij}-t\cdot\Id[i=1]\Id[j=1]$, 
  so its norm square can be approximated to multiplicative error $g$ using one call to $\O$.
  To find an estimate of $\alpha=\frac{\Per_z(X)}{\Per_z(Y)}$, we recursively find complex numbers $\alpha_0,\ldots,\alpha_\ell$ such that the following invariant holds:
  \begin{align}\label{eq:invariant}
    \O(X^{[\alpha_{i+1}]}) \leq 
    \frac{1}{2} \O(X^{[\alpha_i]}) 
  \end{align}
  for every $i$.
  To see how close $\alpha_i$ is to $\alpha$, 
  \begin{align}\nonumber
    |\alpha-\alpha_i| &=\frac{|\Per_z(X)-\alpha_i\Per_z(Y)|}{|\Per_z(Y)|} \\\nonumber
    &\leq \frac{\sqrt{g\O(X^{[\alpha_i]})}}{|\Per_z(Y)|} \\
    &\leq g\sqrt{\frac{\O(X^{[\alpha_i]})}{\O(Y)}}
  \end{align}
  Also let
  \begin{align}
    \beta_i := g\sqrt{\frac{\O(X^{[\alpha_i]})}{\O(Y)}}.
  \end{align}
  Setting the initial guess $\alpha_0=0$, 
  \begin{align}\nonumber
    \O(X^{[\alpha_0]})
    &\leq g |\Per_z(X)|^2 \\\nonumber
    &\leq g|\Per(X)|^2 \\
    &= g(n!)^2.
    \end{align}
    Recall that $\O(Y)\geq g^{-1}|\Per_z(Y)|^2\geq \frac{1}{gM^2}$.
    Then we have 
    \begin{align}\nonumber
      |\alpha-\alpha_i| 
      &\leq g\sqrt{\frac{\O(X^{[\alpha_i]})}{\O(Y)}} \\\nonumber
      &\leq g\sqrt{\frac{2^{-i}\O(X^{[\alpha_0]})}{\O(Y)}} \\\nonumber
      &\leq g \sqrt{2^{-i}\cdot g(n!)^2\cdot g M^2} \\
      &\leq g^2 \cdot n! \cdot M \cdot 2^{-i/2}. 
    \end{align}

    Our goal is to get $|\alpha-\alpha_\ell|\leq \frac{\epsilon}{2n\cdot n! M}$ since it implies that $|\alpha-\alpha_\ell|\leq \frac{\epsilon|\alpha|}{2n}$.
    Thus it suffices to set $2^{\ell/2}> g^2 \cdot (n!)^2\cdot n \cdot M^2\cdot (1/\epsilon)$ and $\ell=O(mn\log n+\log g+\log(1/\epsilon))$.

  Now we give the algorithm that finds an $\alpha_{i+1}$ that satisfies \eq{invariant} from $\alpha_i$.
  By definition, $|\alpha-\alpha_i|\leq \beta_i$.
  We consider the disk
  \begin{align}
    S_i := \{t:|t-\alpha_i|\leq \beta_i\},
  \end{align}
  and a grid of $O(L^2)$ points inside the disk for integer $L$ determined later:
  \begin{align}
    L_i := \left\{\alpha_i+(j-1)\frac{\beta_i}{L}+ \sqrt{-1}(k-1)\frac{\beta_i}{L} \in S_i: |j|,|k|\in[L]\right\}.
  \end{align}
  The algorithm computes $\O(X^{[t]})$ for every $t\in L_i\cap S_i$ and outputs $t$ for which $\O(X^{[t]})$ is minimized
(with ties broken arbitrarily).
  
  By definition, $\alpha\in S_i$, and thus there exists $t\in L_i$ such that $|t-\alpha|\leq \frac{\beta_i}{\sqrt 2L}$.
  For such $t$,
  \begin{align}\nonumber
    \O(X^{[t]})
    &\leq g|\Per_z(X)-t\Per_z(Y)|^2 \\\nonumber
    &\leq g |\alpha-t|^2 |\Per_z(Y)|^2 \\\nonumber
    &\leq g \cdot \frac{\beta_i^2}{2L^2} \cdot |\Per_z(Y)|^2 \\\nonumber
    &\leq g \cdot \frac{\beta_i^2}{2L^2} \cdot g\O(Y) \\
    &= \frac{g^4}{2L^2} \O(X^{[\alpha_i]}).
  \end{align}
  Thus for $L=g^{2}$, the ratio $\frac{\O(X^{[t]})}{\O(X^{[\alpha_i]})}$ is at most $1/2$.
\end{proof}

As a corollary, we recover a result of Aaronson and Arkhipov for $1$-permanent (with slightly worse query complexity) by \thm{we to wa}, but the same theorem can apply to a primitive root of unity $z\in\mathbb T$.
\begin{corollary}[{cf. \cite[Theorem~28]{aaronson2011computational}}]
For $g\in[1,\poly(n)]$, if there is an algorithm $\O$ which on input $X\in\bit^{n\times n}$ approximates $|\Per(X)|^2$ to within multiplicative factor of $g$, 
then there is an algorithm which exactly computes $\Per(X)$ using $O(mn^2g^4\log n+ng^4\log g)$ queries to $\O$.
\end{corollary}
\begin{proof}
  Since $\Per(X)$ is an integer, we apply \thm{we to wa} with $\epsilon=1/n!$.
  The algorithm takes $O(mn^2g^4\log n + ng^4\log g)$ queries.
\end{proof}
Furthermore, obtaining an approximation to additive error within $2^{-n^2}$ takes $\poly(n)$ queries.
\begin{corollary}\label{cor:mul-ad}
For $g\in[1,\poly(n)]$, if there is an algorithm $\O$ which on input $X\in\bit^{n\times n}$, approximates $|\Per_{\zeta_{m}}(X)|^2$ to within a multiplicative factor of $g$, then there exists an algorithm which outputs an esatimate $\tilde P$ of $\Per_{\zeta_{m}}(X)$ satisfying
  \begin{align}
    |\Per_{\zeta_{m}}(X)-\tilde P| \leq 2^{-n^2}
  \end{align}
using $O(mn^2g^4 \log n+ng^4 \log n+g^4 n^3)$ queries to $\O$.
 \end{corollary}
 \begin{proof}
To approximate to additive error within $2^{-n^2}$, it suffices to approximate to multiplicative error $\epsilon=\frac{1}{n!2^{n^2}}$ since the $|\Per_{\zeta_{m}}(X)|$ norm is at most $n!$, and then by the \thm{we to wa} and it takes $O(mn^2g^4 \log n+ng^4 \log n+ng^4 \log(\frac{1}{\epsilon}))=O(mn^2g^4 \log n+ng^4 \log n+g^4 n^3)$ queries.
\end{proof}

Here we see that the proof of approximate hardness relies on the fact that $z$-permanent of given matrix $X$ can be computed from the $z$-permanent of its submatrix (see equation \eq{expan}). 
Other matrix functions such as immanants do not have this property.

\subsection{A Search-to-Decision Reduction}\label{sec:stod}

In the previous section, we have shown how to approximate $\Per_{\zeta_m}(X)$ to multiplicative error $2^{-n^{O(1)}}$ for an $m$-th primitive root of unity $\zeta_m$ using an oracle that approximates $|\Per_{\zeta_m}(X)|^2$ to multiplicative error $g=n^{O(1)}$.
In this section, we show how to search for a representation of $\Per_{\zeta_m}(X)$ from the approximation (see \dfn{exact-representation}). 

Recall that by \eq{alpha-poly}, for binary matrix $X$, $\Per_{\zeta_m}(X)$ is an element in the ring $\mathbb Z[\zeta_m]$ and can be written as a polynomial in $\zeta_m$.
Since $\Phi_m(\zeta_m)=0$ and $\deg\Phi_m=\phi(m)$, we can represent any algebraic integer in $\mathbb Z[\zeta_m]$ as a polynomial of degree at most $\phi(m)-1$.
In fact, such a representation is unique. A more detailed background can be seen \dfn{cyclo}.
\begin{definition}[{Representation of $\mathbb Z[\zeta_m]$}]\label{dfn:general-n}
For integer $m\geq 1$, the representation of an algebraic integer $\alpha\in\mathbb Z[\zeta_m]$ is a tuple $(a_0,\ldots,a_{\phi(m)-1})$ such that 
\begin{align}
    \alpha = \sum_{i=0}^{\phi(m)-1} a_i \zeta_m^i.
\end{align}
As a polynomial in $\zeta_m$, the representation is the remainder of $\alpha$ divided by the cyclotomic polynomial $\Phi_m$.
A representation is said to be bounded in $[A,B]$ for integers $A,B$ if for every $0\leq i\leq \phi(m)-1$, $A\leq a_i \leq B$.
\end{definition}
To extract a representation from an approximation, we rely on the following lemma which states that the coefficients $a_i$ can be represented using $n^{O(1)}$ bits.
\begin{theorem}[Bounds on the coefficients of cyclotomic polynomials \cite{vaughan1975bounds}]\label{thm:cyclotomic-coefficient-bound-2}
For integer $m\geq 1$, let $\Phi_m(x)$ be the $m$-th cyclotomic polynomial.
The maximum absolute value of the coefficients of $\Phi_m(x)$ is at most 
$e^{\frac{1}{2}d(m)\log(m)}$,
where $d(n)$ is the divisor function which count the number of divisors of $m$ (including $1$ and $m$).\footnote{For instance $d(6)=4$.} 
\end{theorem}
For $m=p_1^{a_1}\ldots, p_k^{a_k}$, $d(m)= (a_1+1)\ldots (a_k+1)$.
Thus we can get the trivial bound $d(m)\leq m$.
By \thm{cyclotomic-coefficient-bound-2}, we can represent $\phi_{m}(x)$ using at most only $O(m\log m)$ bits.
Then the following lemma shows that every element of the representation of $\Per_{\zeta_m}(X)$ can be represented using $\poly(n)$ bits.
\begin{lemma}[{\cite[Lemma 2]{bini1986polynomial}}]\label{lem:quotient}
Let $s(x)=\sum_{i=0}^{m}s_ix^i, t(x)=\sum_{i=0}^{n}t_ix^i,$ be two polynomials of degrees $m\geq n$. 
Also let the quotient $q(x)=\sum_{i=0}^{m-n} q_i x^i$ and the remainder $r$ be such that $s(x)=q(x) t(x) + r(x)$.
Then
\begin{equation}
    \sum_{i=0}^{m-n}|q_i| \leq \left(1+\frac{N}{t_n}\right)^{m-n}\sum_{i=n}^{m}\frac{|s_i|}{t_n},
\end{equation}
where $N=\max(|t_n|,|t_{n-1}|,...,|t_{2n-m}|)$ and $t_g=0$ if $g<0.$
\end{lemma} 
For our purpose, since each $A_i$ satisfies 
\begin{align}
    A_\ell = \sum_{\sigma:\ell(\sigma)=\ell} \prod_{i=1}^n X_{i,\sigma(i)} \leq n!, 
\end{align}
by \lem{quotient}, we have the following corollary.
\begin{corollary}\label{cor:per-rep}
  For integer $m\geq 1$ and binary $X$, let $(a_0,\ldots,a_{\phi(m)-1})$ be the representation of $\Per_{\zeta_m}(X)$.
  Then 
  \begin{align}
      \max_{0\leq i\leq \phi(m)-1} \log |a_i| \leq O(m^2\log m +n \log n).
  \end{align}
\end{corollary}
\begin{proof}
 Let 
 \begin{align}
     s(x) := \sum_{\ell=0}^{m-1} A_\ell \zeta_m^\ell.
 \end{align}
 By \thm{cyclotomic-coefficient-bound-2}, the coefficients of the cyclotomic polynomial is at most $M=2^{O(m\log m)}$.
 By \lem{quotient} and the fact that $\Phi_m$ is monic, the summation of the coefficients of the quotient satisfies that 
 \begin{align}\nonumber
     \sum_{i=0}^{m-1-\phi(m)} |q_i|
     &\leq \left(1+M\right)^{m-\phi(m)-1} \cdot n! \cdot (m-\phi(m)-1) \\
     &\leq 2^{O((m-\phi(m))\cdot m\log m + n\log n)}.
 \end{align}
 The representation is the remainder of $s$ devided by $\Phi_m$.
 Let $r(x)=\sum_{i=0}^{\phi(m)-1} a_i x^i= s(x)-q(x)\Phi_m(x)$.
 We have 
 \begin{align}\nonumber
     |a_i| 
     &\leq \left| s_i - \sum_{j=0}^i q_j t_{i-j} \right| \\\nonumber
     &\leq \max_i |s_i| + (\phi(m)-1)\cdot \max_i |q_i| \cdot \max_i |t_i| \\
     &\leq n! + (\phi(m)-1) \cdot 2^{O((m-\phi(m))\cdot m\log m + n\log n)} \cdot M.
 \end{align}
 Thus 
 \begin{align}
     \max_i \log |a_i| \leq O((m-\phi(m))\cdot m\log m +n \log n).
 \end{align}
 Since $\phi(m)=\Omega(\frac{m}{\log\log m})$, $m-\phi(m)=\Theta(m)$.
\end{proof}

Let $A$ denote the maximum of the coefficients in the representation of $\Per_{\zeta_m}(X)$.
By \cor{per-rep}, $\log A=O(m^2\log m+n\log n)$, and thus the coefficients can be represented using $\poly(n)$ bits, provided $m=\poly(n)$.

By Liouville's approximation theorem (\thm{liouville}), the distance of two arbitrary algebraic integers whose representation is bounded in $[-A,A]$ are at least $\delta=2^{-n^{O(1)}}$ in the complex plane.

\begin{corollary}\label{cor:distance-bound}
  For integer $A$, let $\alpha,\beta$ be two algebraic integers in $\mathbb Z[\zeta_m]$ whose representation is bounded in $[-A,A]$.
  Then it holds that $|\alpha-\beta|\geq \frac{1}{(2A\cdot\phi(m)+1)^{\phi(m)-1}}$.
\end{corollary}
\begin{proof}
  Since $\mathbb Z[\zeta_m]$ is an additive group, $\alpha-\beta\in\mathbb Z[\zeta_m]$ whose minimal polynomial is monic and whose representation is bounded in $[-2A,2A]$.
  By \thm{liouville} with $(a,b)=(0,1)$, $d\leq [\mathbb Q(\zeta_m):\mathbb Q]=\phi(m)$, and $|\overline{\alpha-\beta}|\leq Ad$. 
  These facts give
  \begin{align}
      |\alpha-\beta|
      \geq \frac{1}{(2A\cdot\phi(m)+1)^{\phi(m)-1}}.
  \end{align}
\end{proof}

Now let $M:=(2A\phi(m)+1)^{\phi(m)-1}$.
If we already know a complex number $\tilde\alpha$ that is $\delta$-close to an algebraic integer $\alpha$ for $\delta\leq\frac{1}{4M}$, then for every other algebraic integer $\beta\neq\alpha$ whose representation is also bounded in $[-A,A]$, $|\beta-\tilde\alpha|>\frac{3}{4M}$.
Given the observation, we define the following problem.

\begin{problem}[Close algebraic integer problem ($\CAIP$)]\label{prob:alg}
The close algebraic integer problem $(\CAIP)$ is defined as follows:
given $(\alpha,L_0,\ldots,L_{\phi(m)-1},T)$, where $\alpha\in\mathbb C$, $L_0,\ldots,L_{\phi(m)-1},M$ 
are non-negative integers (all represented using $n^{O(1)}$ bits), determine if
\begin{itemize}
\item there is an algebraic integer $\alpha'\in S$ such that $|\alpha-\alpha'|\leq 1/T$, or

\item every $\alpha'\in S$ satisfies $|\alpha-\alpha'|\geq 2/T$,

\end{itemize}
where the set $S$ consists of algebraic integers in $\mathbb Z[\zeta_m]$ whose representation $(a_0,\ldots,a_{\phi(m)-1})$ satisfies $0\leq a_i \leq L_i$ for $0\leq i\leq \phi(m)-1$.
It is promised that the given instance is in one of the cases.
\end{problem}

Indeed, $\CAIP$ can be solved using a non-deterministic Turing machine (NTM).
\begin{theorem}\label{thm:caip-np}
$\CAIP\in\NP$.
\end{theorem}
\begin{proof}
We can use sufficiently many bits (but still polynomial) to represent $\zeta_m$ as a complex number $\xi$ such that   \begin{align}
  |\zeta_m-\xi| = o\left(\frac{1}{T\cdot \phi(m)^2 \cdot \max_i L_i}\right).
  \end{align}
  For every $\beta=b_0+b_1\zeta_m+\ldots+b_{\phi(m)-1}\zeta_{m}^{\phi(m)-1}$  approximated by $\tilde\beta=b_0+b_1\xi+\ldots+b_{\phi(m)-1}\xi^{\phi(m)-1}$  satisfying $|\beta-\tilde\beta|=o(1/T)$ for every $i$ and integer $b_i\in[L_i]$.

If the NTM non-deterministically finds $b_0,\ldots,b_{\phi(m)-1}$ satisfying $|\alpha-\tilde\beta|\leq 1.5/T$, then $|\alpha-\beta|\leq 1/T$.
Otherwise, every $\beta$ satisfies $|\alpha-\beta|\geq 2/T$.
\end{proof}

The problem $\CAIP$ only consider the case where each coefficient $b_i\in[0,L_i]$. 
To extend it to general cases where the $i$-th element is in $[U_i,L_i]$, we can shift $\alpha$ to $\beta=\alpha-U_i \zeta_m^i$ and solves $\CAIP$ with $\beta$ and the $i$-th paremeter being $L_i-U_i$.
Thus it suffices to only consider the former setting.

Given an oracle for $\CAIP$, we can search for a satisfying algebraic integer.

\begin{theorem}\label{thm:search-caip}
  Let $\O$ be an oracle which solves $\CAIP$. 
  Then there exists an algorithm which makes $O(\log(L_1\ldots L_{\phi(m)-1}))$ queries to $\O$ and finds a satisfying algebraic integer, or returns $\bot$ if there is not.
\end{theorem}
\begin{proof}
  The algorithm $\A$ is basically a binary search for each coefficient.

  First $\A$ runs $b\gets\O(\alpha,L_0,\ldots,L_{\phi(m)-1},T)$; if the output $b=0$, then $\A$ returns $\bot$.
  Otherwise, there must be a satisfying algebraic integer $\alpha'$.
  Then $\A$ runs the following steps:
  \begin{enumerate}
    \item Set $\beta\gets\alpha$. 
    \item For $j=\phi(m)-1,\ldots, 0$, 
  \begin{enumerate}
  \item Set $W\gets 0$ and $R\gets \lfloor L_j/2\rfloor$.
  \item Repeat the following steps until $R=0$:
    \begin{enumerate}
    \item Compute $b_i\gets\O(\beta,L_0,\ldots, L_{j-1}, R, 0,\ldots, 0 ,T)$.
    \item If $b=0$, set $\beta\gets\beta-(R+1) \xi^{j}$ and $W\gets W+R+1$.\label{stp:run}
    \item Set $R\gets\lfloor R/2\rfloor$.
    \end{enumerate}
  \item Set $a_{j}\gets W$.
  \end{enumerate}
  \item Output $(a_0,\ldots,a_{\phi(m)-1})$.
  \end{enumerate}
  We prove the algorithm is correct for $j=\phi(m)-1$; for other coefficients, the correctness analysis applies similarly.
  First, the algorithm must terminate since in each iteration $R$ gets decreased by a factor of two. 
  For the correctness, we note that $\beta + W z^{\phi(m)-1}=\alpha$ holds after each iteration.
  In each step, let a satisfying algebraic integer denote $\beta'$, written as
  \begin{align}
    \beta' = c_{\phi(m)-1} z^{\phi(m)-1}+a_{\phi(m)-2}z^{\phi(m)-2}+\ldots + a_0
    \end{align}
    before \stp{run} is run. 
    It suffices to show that after \stp{run}, the leading coefficient is no more than $R$ by induction.
    For the base case, since the leading coefficient of $\alpha'$ is no more than $L_{\phi(m)-1}$; the assertion holds.
    If $b=0$, then $R< c_{\phi(m)-1}\leq 2R+1$, and $c_{\phi(m)-1}-(R+1)\geq 0$; otherwise, there exists a satisfying $\beta'$ such that $c_{\phi(m)-1}\leq R$.
    Thus after \stp{run}, there is a satisfying algebraic integer whose leading coefficient is no more than $R$.
\end{proof}

\section{Average-case Hardness}\label{sec:wa}

In the context of average-case hardness, Aaronson and Arkhipov~\cite{aaronson2011computational} considered the hardness of $1$-permanent for i.i.d. Gaussian matrices.
More recently, Haferkamp, Hangleiter, Eisert, and Gluza \cite{haferkamp2020contracting} considered the average-case hardness of the truncated uniform distribution.

In this section, we extend their results both to $z$-permanent and to a wide range of distributions, provided that the distribution satisfies a property called the \emph{strong autocorrelation property}.

Our idea follows from Aaronson and Arkhipov \cite{aaronson2011computational}: 
If one can compute a low-degree polynomial for sufficiently many points, then there is an algorithm that recovers the whole polynomial with high probability. 
More precisely, we use the following algorithm for noisy polynomial interpolation. 
\begin{theorem}[{Berlekamp-Welch algorithm \cite{bakaspolynomial}}]\label{thm:bwa} 
Let $q$ be a univariate polynomial of degree $d$ over any field
$\mathbb{F}$. Suppose we are given m pairs of $\mathbb{F}$ elements $(x_1, y_1),...,(x_m, y_m)$ (with the $x_i's$ all distinct), and are
promised that $y_i = q(x_i)$ for more than $(m+d)/2$ values of i. Then there is a deterministic algorithm to
reconstruct $q$ using $\poly(n,d)$ field operations. 
\end{theorem}

Now we define the following property, under which a random self-reducibility using the Berlekamp-Welch algorithm can be applied.
\begin{definition}[Strong autocorrelation property]\label{dfn:autocor}
Consider a random variable $X$ with a real distribution $\mathcal{F}$ with mean $0$ and variance $1.$ We denote it as $X \sim \mathcal{F}_{\mathbb{R}}(0,1).$ We define $\mathcal{F}_{\mathbb{R}}(\mu,\sigma^2)$ with mean $\mu$ and variance $\sigma^2$ by shifting and scale the random variable $X \to \mu+\sigma^2 X.$ We pick $\epsilon<\frac{1}{2}.$ We can define complex random variable $Z=X+iY$,$X \sim \mathcal{F}_{\mathbb{R}}(\mu_1,1/2), Y \sim \mathcal{F}_{\mathbb{R}}(\mu_2,1/2)$ and $X,Y$ are independent. We write $Z \sim \mathcal{F}_{\mathbb{C}}(\mu_1+i \mu_2,1).$  Now we consider the following three distributions
\begin{align}
X & = \F_{\mathbb{C}}(0,1)^{n} \\ 
\mathcal{D}_{1}  &  =\F_{\mathbb{C}}\left(0,\left(  1-\varepsilon\right)
^{2}\right)^{n},\\
\mathcal{D}_{2}  &  =\prod_{i=1}^{n}\mathcal{F}_{\mathbb{C}}\left(  v_{i},1\right)
\end{align}
for some vector  $v\in\mathbb{C}^{n}$. 
A distribution is said to satisfy the strong autocorrelation property \dfn{autocor} if there exists constant $M_1, M_2$ independent of $n, \epsilon$, such that
\begin{align}
\left\Vert \mathcal{D}_{1}-X\right\Vert  &  \leq2n M_1\varepsilon,\\
\left\Vert \mathcal{D}_{2}-X\right\Vert  &  \leq \sqrt{2n}M_2 \left\Vert
v\right\Vert_{2}\,.
\end{align}
\end{definition}
When $\F$ is complex Gaussian or truncated uniform distribution both satisfy this criterion from direct calculation \cite{aaronson2011computational,haferkamp2020contracting}.
In \app{strong}, we give a sufficient condition under which almost all well-behaved $\F$ satisfies the such criterion.

Now we are ready to show that, as permanent, $z$-permanent also has random self-reducibility. 
\begin{theorem}[Random self-reducibility of  $z$-permanent]\label{thm:random} 
Assume that there is no classical efficient algorithm for computing $z$-permanent 
and for all $\{0,1\}$-matrices.
For all $\delta\geq1/\poly(n) $ the following problem is classically hard: Given an $n\times n$ matrix $X$ drawn from
$\mathcal{F}_{\mathbb{C}}(0,1)^{n \times n}$ which satisfies the strong autocorrelation property, output $\Per_z\left( X\right)
$ with probability at least ${3}/{4}+\delta$ over $\mathcal{F}_{\mathbb{C}}(0,1)^{n \times n}$.
\end{theorem}
\begin{proof}
Let $M=\left(  x_{ij}\right)  \in\left\{  0,1\right\}  ^{n\times n}$ be an
arbitrary $\{0,1\}$-matrix. We show how to compute $\Per_z
\left(  M \right) $ in probabilistic polynomial time, given access to an
oracle $\mathcal{O}$ such that
\begin{equation}
\Pr_{X\sim \mathcal{F}_{\mathbb{C}}(0,1)^{n\times n} }\left[  \mathcal{O}\left( X\right)
=\Per_z\left(  X\right)  \right]  \geq\frac{3}{4}+\delta\,.
\end{equation}
We choose a matrix $Y\in\mathbb{C}^{n\times n}$ from the distribution $\mathcal{F}_{\mathbb{C}}(0,1)^{n\times n}$ and define
\begin{equation}
X\left(  t\right)  :=\left(  1-t\right)  Y+tM, q\left(  t\right)  :=\Per_z\left(  X\left(  t\right)  \right)
\end{equation}
so that $X\left(  0\right)  =Y$ and $X\left(  1\right)  =M $, and $q\left(  t\right) $ is a univariate polynomial in $t$ of degree at
most $n$ and $q\left(  1\right)  =\Per_z\left(  X\left(
1\right)  \right)  =\Per_z\left(  M\right)  $.
Now let 
\[
L:=\left\lceil \frac{n}{\delta}\right\rceil\quad\text{and}\quad \varepsilon:=
\frac{\delta}{\left(  4n^{2}M_1+4n^2M_2\right)  L}\,.
\]
 For each $\ell\in\left[  L\right]
$ call the oracle $\mathcal{O}$ on input matrix $X\left(  \varepsilon
\ell\right)  $. Then, using the Berlekamp-Welch algorithm (\thm{bwa}), attempt to find a degree-$n$ polynomial $q^{\prime}
:\mathbb{C}\rightarrow\mathbb{C}$ such that
\begin{equation}
q^{\prime}\left(  \varepsilon\ell\right)  =\mathcal{O}\left(  X\left(
\varepsilon\ell\right)  \right)
\end{equation}
for at least a ${3}/{4}+\delta$ fraction of $\ell\in\left[  L\right]  $.
\ If no such $q^{\prime}$ is found, then fail; otherwise, output $q^{\prime
}\left(  1\right) $ as the\ guessed value of $\Per_z\left(
M\right)  $.
Now we show that the above algorithm succeeds (that is, outputs $q^{\prime
}\left(  1\right)  =\Per_z\left(  M\right)  $) with probability
at least ${1}/{2}+ {\delta}/{2}$ over $Y$. Provided that holds, it is
clear that the success probability can be boosted to (say) $2/3$ by simply
repeating the algorithm $O\left(  1/\delta^{2}\right) $ times with different
choices of $Y$ and then outputting the majority result.
To prove the claim, note that for each $\ell\in\left[  L\right] $ one can
think of the matrix $X\left(  \varepsilon\ell\right) $ as having been drawn
from the distribution
\begin{equation}
\mathcal{D}_{\ell}:=\prod_{i,j=1}^{n}\F_{\mathbb{C}}\left(  \varepsilon\ell
a_{ij},\left(  1-\varepsilon\ell\right)  ^{2}\right)\,.
\end{equation}
Let
\begin{equation}
\mathcal{D}_{\ell}^{\prime}:=\prod_{i,j=1}^{n}\F_{\mathbb{C}}\left(
\varepsilon\ell a_{ij},1\right).
\end{equation}
Then by the triangle inequality together with \dfn{autocor} ($X \sim  \F_{\mathbb{C}}(0,1)^{n\times n}$),
\begin{align}\nonumber
\left\Vert \mathcal{D}_{\ell}-X\right\Vert  &
\leq\left\Vert \mathcal{D}_{\ell}-\mathcal{D}_{\ell}^{\prime}\right\Vert
+\left\Vert \mathcal{D}_{\ell}^{\prime}-X\right\Vert \\\nonumber
&  \leq2n^{2}M_1\varepsilon\ell+\sqrt{2n^2}M_2\sqrt{n^{2}\left(  \varepsilon\ell\right)^{2}
}\\\nonumber
&  \leq\left(  2n^{2}M_1+2n^2 M_2\right)  \varepsilon l \leq (2n^2 M_1+ 2n^2 M_2) \varepsilon L\\
&  \leq\frac{\delta}{2}\,.
\end{align}
We have
\begin{align}\nonumber
\Pr\left[  \mathcal{O}\left(  X\left(  \varepsilon\ell\right)  \right)
=q\left(  \varepsilon\ell\right)  \right]   &  \geq\frac{3}{4}+\delta
-\left\Vert \mathcal{D}_{\ell}-X\right\Vert \\
&  \geq\frac{3}{4}+\frac{\delta}{2}\,.
\end{align}
Now let $S$ be the set of all $\ell\in\left[  L\right] $ such that
$\mathcal{O}\left(  X\left(  \varepsilon\ell\right)  \right)  =q\left(
\varepsilon\ell\right)  $. Then by the reverse Markov's inequality \cite{eisenberg2001generalization},
\begin{equation}
\Pr\left[  \left\vert S\right\vert \geq\left(  \frac{1}{2}+\frac{\delta}
{2}\right)  L\right]  \geq1-\frac{\frac{1}{4}-\frac{\delta}{2}}{\frac{1}
{2}-\frac{\delta}{2}}\geq\frac{1}{2}+\frac{\delta}{2}\,.
\end{equation}
So we then just run the above algorithm $O(\frac{1}{\delta^2})$ times to amplify the success probability. After that, with high probability, we have
\[
\left\vert S\right\vert \geq\left(  \frac{1}{2}
+\frac{\delta}{2}\right)  L\,.
\]
Then by \thm{bwa}, the
Berlekamp-Welch algorithm succeeds; It will output polynomial
$q^{\prime}$ will be equal to $q$. This proves the claim.
\end{proof}

\section{Average-case Hardness of Approximation}\label{sec:van}

Since Aaronson and Arkhipov's seminal work on Boson Sampling \cite{aaronson2011computational} appeared,  approximating $1$-permanent in the average case has been investigated for particular distributions \cite{aaronson2011computational,haferkamp2020contracting}. 
In particular, Aaronson and Arkhipov \cite{aaronson2011computational} conjectured a problem called \textsc{Gaussian Permanent Estimation} (also denoted $\GPE_{\times}$) is  $\sharpP$-hard. 
The problem $\GPE_\times$ asks to compute $\Per(X)$ to multiplicative error $\epsilon$ with probability at least $1-\delta$ over matrix $X$ whose each element is sampled independently from the normalized Gaussian distribution. 
For quantum supremacy from Boson Sampling, the ultimate goal was to prove the additive version of $\GPE_\times$, denoted $\GPE_\pm$, is also hard.
Based on another conjecture, called the Permanent Anti-Concentration Conjecture (PACC), one can deduce that $\GPE_\pm$ is as hard as $\GPE_\times$. In the following, we denote $\mathcal{G}=\N_{\mathbb{C}}(0,1)^{n\times n}$ the distribution of an $n \times n$ i.i.d. Gaussian matrix with zero mean and unit variance. 
\begin{problem}[{$\GPE_\times$}]
  Given as input $X\sim\N_{\mathbb{C}}(0,1)^{n\times n}$ of i.i.d. Gaussians (complex normal distribution) and $\epsilon,\delta>0$, estimate $\Per(X)$ to within error $\pm\epsilon|\Per(X)|$ to probability at least $1-\delta$ over $X$.
\end{problem}

\begin{problem}[{$\GPE_{\pm}$}]
  Given as input $X\sim\N_{\mathbb{C}}(0,1)^{n\times n}$ of i.i.d. Gaussians (complex normal distribution) and $\epsilon,\delta>0$, estimate $\Per(X)$ additively (i.e., to within error $\pm \epsilon \sqrt{n!}$), with
probability at least $1-\delta$ over $X$
\end{problem}

In this section, we consider the hardness of $z$-permanent for matrices sampled from a general distribution $\H$ whose each entry is i.i.d. sampled. 
We first define the multiplicative and the additive versions of \textsc{Permanent Estimation} (PE) for distribution $\H$.
\begin{problem}[$\PE_{\H,z,\times}$]
Given as input $X\sim \H$ and $\epsilon,\delta>0$, estimate $\Per(X)$ to within error $\pm\epsilon|\Per(X)|$ to probability at least $1-\delta$ over $X$. 
\end{problem}
For convenience, $\PE_{\H,1,\times}$ is also denoted $\PE_{\H,\times}$. 
Similarly, we can define an additive version of the above problem. 
\begin{problem}[$\PE_{\H,z,\pm}$]
Given as input $X\sim \H$ and $\epsilon,\delta>0$, estimate $\Per_z(X)$ to within error $\pm\epsilon g(n)$ to probability at least $1-\delta$ over $X$.
For convenience, $\PE_{\H,1,\pm}$ is also denoted $\PE_{\H,\pm}$. 
\end{problem}
From the definition, we have $\GPE_{\times}=\PE_{\mathcal G,\times}$ and $\GPE_{\pm}=\PE_{\mathcal G,\pm}$ with $g=\sqrt{n!}$.
Let $d=\binom{n}{2}$
For $r\in(0,\frac{1}{d+1})$, let $S_r:=\{z = e^{i2\pi r}\zeta_{d+1}^{i-1}:i\in[d+1]\}$, where $\zeta_{k}:=e^{i2\pi/k}$ is a primitive $k$-th root of unity. 
Let $\O$ be an oracle, given $X \sim \H$, outputs the approximations of $z$-permanent of $d+1$ points, i.e.,
$\Per_{z_1}(X),\ldots,\Per_{z_{d+1}}(X)$ for $z_1,\ldots,z_{d+1} \in S_r$.
More formally, we define $\O_\times$ and $\O_{\pm}$ to be oracles that approximate to multiplicative and additive errors, respectively:
\begin{definition}[$\O_{\times}, \O_{\pm}$]\label{dfn:Oracle} 
The oracles $\O_\times$ and $\O_{\pm}$ are defined as follows:
\begin{itemize}
\item 
Given $X \sim \H,$ $\O_{\times}(X,\epsilon,\delta)$ is the oracle that outputs a vector $y=(y_1,\ldots,y_{d+1})^\top$ 
such that for every $i \in [d+1]$, $|y_i-\Per_{z_i}(X)|\leq \epsilon |\Per_{z_i}(X)|$ with probability at least $1-\delta$ over choices of $X$. 

\item 
Given $X \sim \H$ and efficiently computable function $g$, $\O_{\pm}(X,\epsilon,\delta)$ is the oracle that outputs a vector $y=(y_1,\ldots,y_{d+1})^\top$ such that for every $i \in [d+1], $ we have: $|y_i-\Per_{z_i}(X)|\leq \epsilon g(n)$ with probability at least $1-\delta$ over choices of $X$.
\end{itemize}
\end{definition}

The rest of this section is devoted to showing that if $\O_\times$ can be classically simulated, then $\PE_{\H,\times}$ can be solved efficiently, and similarly holds for $\O_\pm$ and $\PE_{\H,\pm}$ in place of $\O_\times$ and $\PE_{\H,\pm}$.
The proof is similar to Lipton's proof \cite{lipton1991new} for the average hardness of permanent over a finite field $\mathbb{F}_{p}$ for $p>n$. 
In our case, we can approximate $\Per(X)$ given an approximation of $\Omega(n^2)$ points on the unit circle.

First, we show the following two lemmas for polynomial interpolation. One is for multiplicative error and one is for additive error.
\begin{lemma}\label{lem:additive-reduction}
  For integer $d\geq 1$, let $f$ be a degree-$d$ complex-valued polynomial and $\A$ be an algorithm that approximates $f$ to additive error $\epsilon$ on points $x_1,\ldots,x_{d+1}$. 
There exists an algorithm which, on input $x\in\mathbb{T}$, makes $d+1$ calls to $\A$ and outputs an estimate of $f(x)$ to additive error $(d+1)\|V^{-1}\|\epsilon$, where $V=V(x_1,\ldots,x_{d+1})$ is the Vandermonde matrix with nodes $x_1,\ldots,x_{d+1}$. 
\end{lemma}
\begin{proof}
  Let $\tilde y_1,\ldots,\tilde y_{d+1}$ be the estimate of $\A$ on $x_1,\ldots,x_{d+1}$ respectively.
  Let $\tilde y = (\tilde y_1,\ldots,\tilde y_{d+1})^\top$ and $x=(x_1,\ldots,x_{d+1})^\top$ be the associated column vectors.
  Since for each $i\in[d+1]$, $|\tilde y_i-y_i|\leq \epsilon$, $\|\tilde y-y\|\leq\epsilon\sqrt{d+1}$.

  Let $c=(c_0,\ldots,c_d)^\top$ be the vector of the coefficients of $f(x)=\sum_{i=0}^d c_ix^i$.
  Solving the linear equation $V\tilde c=\tilde y$ yields an estimate of the coefficients.
  For $x\in\mathbb{T}$, let $v=(1,x,x^2,\ldots,x^d)$ be the vector of monomials evaluated on $x$.
  Our goal is to bound the distance 
  \begin{align}
    |v\cdot(\tilde c-c)| \leq \|(V^{-1})^\top v\|\|V(c-c')\|\leq \|v\|\|V^{-1}\|\|\tilde y - y\| \leq (d+1)\|V^{-1}\|\epsilon.
  \end{align}
  The first inequality holds by Cauchy-Schwarz inequality. The second holds by the definition of matrix norm. The last holds because $\|\tilde y-y\|\leq\epsilon$ and $\|v\|=\sqrt{d+1}$.
\end{proof}

If we restrict our problem to approximating $f(1)$ for polynomials with non-negative coefficients, we can do better.
In particular, we show that if $\A$ approximates $x_1,\ldots,x_{d+1}$ for polynomial $f$ with non-negative coefficients to multiplicative error $\epsilon$, then there is an algorithm that approximates $f(1)$ to the multiplicative error as above.
\begin{lemma}\label{lem:multiplicative-reduction}
  For integer $d\geq 1$, let $f$ be a degree-$d$ polynomial with non-negative coefficients and $\A$ be an algorithm that approximates $f$ to multiplicative error $\epsilon$ on $x_1,\ldots,x_{d+1}\in\mathbb T$.
  Then there exists an algorithm which approximates $f(1)$ to multiplicative error $(d+1)\|V^{-1}\|\epsilon$, where $V=V(x_1,\ldots,x_{d+1})$ is the Vandermonde matrix with nodes $x_1,\ldots,x_{d+1}$.
\end{lemma}
\begin{proof}
  Let $\tilde y_1,\ldots,\tilde y_{d+1}$ be the estimate of $\A$ on $x_1,\ldots,x_{d+1}$ and $y=(\tilde y_1,\ldots,\tilde y_{d+1})^\top$.
  Solving the linear equation $V\tilde c=\tilde y$ yields $\tilde c=V^{-1}\tilde y$.
  Also let $\tilde f(z):=\sum_{i=0}^d \tilde c_i z^i$ be the approximated polynomial obtained by solving the linear system.

  Since $f$ has non-negative coefficients, $|f(x)|\leq f(1)=\langle 1^{d+1},c\rangle$ for every $x\in\mathbb T$.
  Since $\|\tilde y-y\|\leq \epsilon \|y\|\leq \epsilon \sqrt{d+1}\cdot f(1)$, 
  \begin{align}
    |\tilde f(1)-f(1)|
    = |\langle 1^{d+1},\tilde c-c\rangle|
    \leq \sqrt{d+1}\|\tilde c- c\|\leq (d+1)\|V^{-1}\| \cdot \epsilon \cdot |f(1)|.
    \end{align}
    The first inequality holds by Cauchy-Schwarz inequality. 
    Thus an algorithm which outputs $|\tilde c|_1$ as the approximation with multiplicative error $(d+1)\epsilon\|V^{-1}\|$.
\end{proof}

Now we are ready to prove \thm{GPE-z}. 
By \lem{additive-reduction} and \lem{multiplicative-reduction}, we separate them into two cases. 
\begin{theorem}\label{thm:GPE-z}
The following statements hold:
\begin{enumerate}
\item Suppose $\PE_{\H,\times}$ is $\sharpP$-hard and $\H$ is non-negative distribution, and the polynomial hierarchy is infinite, then $\O_{\times}$ does not admit a classical efficient algorithm, i.e., there is no simulation of $\O_\times$ running in time $\poly(n,\frac{1}{\epsilon},\frac{1}{\delta})$. 
 
\item Suppose $(\PE_{\H,\pm},g)$ is $\sharpP$-hard and the polynomial hierarchy is infinite, then $\O_{\pm}$ does not admit a classical efficient algorithm, i.e., there is no classical simulation of $\O_{\pm}$ running in time $\poly(n,\frac{1}{\epsilon},\frac{1}{\delta})$).
\end{enumerate}
\end{theorem}
\begin{proof}
For (1), we give non-constructive proof.
  Let $d=\binom{n}{2}$.
  For $r\in(0,\frac{1}{d+1})$, let $S_r:=\{z = e^{i2\pi r}\zeta_{d+1}^{i-1}:i\in[d+1]\}$, where $\zeta_{k}:=e^{i2\pi/k}$ is a primitive $k$-th root of unity.
  We show that the number of points $z\in S_r$ such that $\Per_z(X)$ is efficiently computable cannot be more than $d$.
  Suppose toward contradiction this is not the case.
  Then let $\O_{\times}$ be algorithms that runs in time $\poly(n,1/\epsilon,1/\delta)$ and approximates $\Per_{z_1}(X),\ldots,\Per_{z_{d+1}}(X)$ to precision $\epsilon$ with probability $1-\delta$, for $S_r=\{z_1,\ldots,z_{d+1}\}$.
  Applying the algorithm in the proof of \lem{multiplicative-reduction}, we apply the oracle $\O_{\times}(X,\epsilon,\delta)$ for $\H_\times$ by the following algorithm:
  \begin{enumerate}
  \item Get $y=(y_1,\ldots,y_{d+1})^\top\gets \O_{\times}(X,\epsilon(d+1)^{-1/2},\delta/(d+1))$;
  \item Solving the linear system $Vc=y$, where $V=V(z_1,\ldots,z_{d+1})$.
  \item Output the summation of elements of $c$ as the approximation of $\Per(X)$.
  \end{enumerate}
  By the assumption, for each $i\in[d+1], $with probability at least $1-\frac{\delta}{d+1}$, $|y_i-\Per_{z_i}(X)|\leq \frac{\epsilon}{\sqrt{d+1}}|\Per_{z_i}(X)|$.
  Since points in $S_r$ is equally spaced, $\frac{1}{\sqrt{d+1}}V$ is a unitary matrix, and so is $\sqrt{d+1}V^{-1}$. 
  This implies that the matrix norm $\|V^{-1}\|=(d+1)^{-1/2}$.
  By \lem{multiplicative-reduction} and union bound, the algorithm outputs an approximation of $\Per(X)$ to multiplicative error $\epsilon$ with probability at least $1-\delta$.
  By the assumption $\PE_{\H,\times}$ is hard, we conclude $\P^{\sharpP}\subseteq\BPP^{\O_{\times}}\subseteq\BPP$ and a collapse of the polynomial hierarchy.

For (2), we use the same $S_r$, then let $\O_{\pm}$ be algorithms that runs in time $\poly(n,1/\epsilon,1/\delta)$ and approximates $\Per_{z_1}(X),\ldots,\Per_{z_{d+1}}(X)$ to precision $\epsilon$ with probability $1-\delta$, respectively for $S_r=\{z_1,\ldots,z_{d+1}\}$.
  Applying the algorithm in the proof of \lem{additive-reduction}, we apply the oracle $\O_{\pm}(X,\epsilon,\delta)$ for $\H_\pm$ by the following algorithm:
  \begin{enumerate}
  \item Get $y=(y_1,\ldots,y_{d+1})^\top\gets \O_{\pm}(X,\epsilon  (d+1)^{-1/2},\delta/(d+1))$.
  \item Solving the linear system $Vc=y$, where $V=V(z_1,\ldots,z_{d+1})$.
  \item Output the summation of elements of $c$ as the approximation of $\Per(X)$.
  \end{enumerate}
  By the assumption, for each $i\in[d+1], $with probability at least $1-\frac{\delta}{d+1}$, $|y_i-\Per_{z_i}(X)|\leq \frac{\epsilon}{\sqrt{d+1}}g(n)$.
  Since points in $S_r$ is equally spaced, $\frac{1}{\sqrt{d+1}}V$ is a unitary matrix, and so is $\sqrt{d+1}V^{-1}$. 
  This implies that the matrix norm $\|V^{-1}\|=(d+1)^{-1/2}$.
  By \lem{additive-reduction} and a union bound, the algorithm outputs an approximation of $\Per(X)$ to additive error $\epsilon g(n)$ with probability at least $1-\delta$.
  By our assumption $(\PE_{\H,\times},g)$ $\sharpP$-hard, we conclude that $\P^{\sharpP}\subseteq\BPP^{\O_{\pm}}\subseteq\BPP$ and the collapse of the polynomial hierarchy.
\end{proof}

At first glance, it seems like our result does not work for an i.i.d. Gaussian matrix since it is not positive except with exponentially small probability. 
However, similarly as in \cite{aaronson2011computational}, we prove that $\PE_{\mathcal{G},z,\pm}$ with $g(n)=\sqrt{n!}$ can be reduced to $\PE_{\mathcal{G},z,\times}$.

\begin{lemma}\label{lem:pm-to-times}
$(\PE_{\mathcal{G},z,\pm})$ with $g=\sqrt{n!}$ is
polynomial-time reducible to $\PE_{\mathcal{G},z,\times}$.
\end{lemma}
\begin{proof}
  Suppose we have a polynomial-time algorithm $M$ outputs a good multiplicative approximation to $\Per_z(X)$—that is, a $w$ such that:
  \begin{equation}
      |w-\Per_z(X)| \leq \epsilon |\Per_z(X)|.
  \end{equation}
with probability at least $1-\delta$ over $X$ by Markov’s inequality ($\Exp_{X\sim\G}[|\Per_z(X)|^2] = n!,$ see \lem{first-m}):
 \begin{equation}
\Pr_{X\sim\G}\left[|w-\Per_{z}(X)|>k \sqrt{n!}\right]<\frac{1}{k^2}
  \end{equation}
by the union bound,
\begin{align*}
    \Pr_{X\sim\G}\left[|w-\Per_{z}(X)|>\epsilon k \sqrt{n!}\right] 
    &\leq \Pr_{X\sim\G}\Big[|w-\Per_z(X)|>\epsilon |\Per_z(X)|\Big]+ \Pr_{X\sim\G}\Big[\epsilon|\Per_{z}(X)|>\epsilon k \sqrt{n!}\Big]\nonumber \\
    &\leq \delta+\frac{1}{k^2}.
\end{align*}

Letting $\delta=\frac{\delta'}{2}, k=\sqrt{\frac{2}{\delta}}$ and $\epsilon=\frac{\epsilon'}{k}$, we have $\delta+\frac{1}{k^2}=\delta'$ which finishes the proof.
\end{proof}

We now introduce the conjecture that generalize the original permanent anti-concentration conjecture. 

\begin{conjecture}[$z$-permanent anti-concentration conjecture ($z$-PACC)]\label{conj:pacc}
  There exists a polynomial $p$ such that for positive integer $n$, real number $\delta>0$ and $z\in\mathbb T$, 
  \begin{align}
    \Pr_{X\sim\G}\left[ |\Per_z(X)|^2 \geq \frac{\sqrt{n!}}{p(n,1/\delta)}\right] \geq 1-\delta. 
  \end{align}
\end{conjecture}
In particular, $z=-1$ is proved by Aaronson \cite{aaronson2011computational} and $z=1$ is standard conjecture that widely believed to be True.

\begin{corollary}
Assume that the $1$-PACC holds. 
Suppose $\GPE_{\times}$ is $\sharpP$-hard and the polynomial hierarchy is infinite, then $\O_{\times}$ does not admit a classical efficient algorithm, i.e., there is no simulation of $\O_{\times}$ running in time $\poly(n,\frac{1}{\epsilon},\frac{1}{\delta})$.
\end{corollary}
\begin{proof}
    Suppose $\PE_{\mathcal{G},z,\times}$ admit a classical efficient algorithm in $\poly(n,\frac{1}{\epsilon},\frac{1}{\delta})$.  
    Then one can solve $\PE_{\mathcal{G},z,\times}$ then solving $\GPE_{z,\times}$ for $d+1$ points. 
    By \lem{pm-to-times}, it solves $(\PE_{\mathcal{G},z,\pm})$ with $\sqrt{n!}$ for $d+1$ points which solves $\GPE_{\pm}$. Assuming $1$-PACC \cite{aaronson2011computational}, $\GPE_{\pm}$ is polynomial equivalent to $\GPE_{\times}.$ Hence there is classical polynomial solves $\GPE_{\times}$ which contradicts our assumption that polynomial hierarchy is infinite.
\end{proof}

Finally, we show that approximating $\Per_z(X)$ is as hard as approximating $\Per_{z^*}(X)$ for any distribution $X$ such that $X,X^*$ are identically distributed. A similar proof idea applies to $|\Per_z(X)|^2$ as well.
  \begin{lemma}
   For any distribution $X$ such that $X, X^*$ are identically distributed, if there is an algorithm running in time $\poly(n,1/\delta,1/\epsilon)$ for approximating $\Per_z(X)$ to within multiplicative (resp. additive) error $\epsilon$ with probability $1-\delta$, then there is an algorithm which approximates $\Per_{z^*}(X)$ to within multiplicative (resp. additive) error $\epsilon$ with probability $1-\delta$.
  \end{lemma}
  \begin{proof}
    By the assumption, let $\A(X,0^{1/\epsilon},0^{1/\delta})$ be an algorithm that approximates $\Per_z(X)$ to error $\epsilon$ with probability $1-\delta$. 

    Since $\Per_{z^*}(X)^* = \Per_z(X^*)$ and $X,X^*$ are identically distributed, for computing $\Per_{z^*}(X)$, the algorithm $\B$ runs $\A(X^*,0^{1/\epsilon},0^{1/\delta})$ to obtain an approximation $\tilde P$ of $\Per_z(X^*)$ with probability $1-\delta$ over $X^*$.
    Then since the absolute value of a complex number does not change by taking complex conjugate, 
    \begin{align}\nonumber
      \Pr_X[|\tilde P - \Per_{z}(X^*)|\leq \alpha(z,X^*) ] 
      &= \Pr_X[|\tilde P^* - \Per_{z^*}(X)|\leq \alpha(z^*,X) ] \\
      &\geq 1-\delta.
    \end{align}
    Here for multiplicative error 
    \begin{align}
    \alpha(z,X^*)=\epsilon |\Per_z(X^*)| = \epsilon |\Per_{z^*}(X)| = \alpha(z^*,X).
      \end{align}
    For additive error, 
    \begin{align}
      \alpha(z,X^*) 
      = \epsilon\sqrt{n!}
      = \alpha(z^*,X).
    \end{align}
  \end{proof}

\section{Open Problems and Discussion}\label{sec:final}

In this paper, we studied the computational complexity of quantum determinants, a $q$-deformation of matrix permanents. 
The $q$-permanent of a matrix $X$ is defined and it generalizes both determinant and permanent. It is shown that computing the $q$-permanent is $\mathsf{Mod}_p\mathsf{P}$-hard for a primitive $m$-th root of unity $q$ for odd prime power $m=p^k$. This result implies that an efficient algorithm for computing $q$-permanent would result in a collapse of the polynomial hierarchy. 
We also showed that an efficient approximation algorithm for $q$-permanent would also imply a collapse of the polynomial hierarchy. The hardness of computing $q$-permanent remains to hold for a wide range of distributions by random self-reducibility. The techniques developed for the hardness of permanent can be generalized to $z$-permanents for both worst and average cases.

Here we include a few interesting questions that are not solved in this paper. 
\begin{itemize}
  \item 
  It is known that there is a polynomial-time algorithm for determining whether $\Per(X)\neq 0$ for $\{0,1\}$-matrix $X$ by determining if there is a perfect matching in a bipartite graph whose adjacent matrix is $X$.  
  Is the property about $z$-permanent efficiently computable for $z\neq 1$, i.e., is there a polynomial-time algorithm for determining whether $\Per_z(X)\neq 0$? 
  
  \item We have shown that \prob{alg} is in $\NP$. 
  Does \prob{alg} have polynomial-time algorithm? 
  
  \item 
  In this work, we have discussed the computational complexity of $z$-permanent for $z\in\mathbb T$.
  It is natural to consider other subsets of the complex numbers.
  For example, what is the computational complexity of $z$-permanent for real $z$?
  
  \item As explained in \sec{prime-power}, the techniques we used for proving the hardness of $\zeta_m$-permanent can only work for the special case where $m$ is a prime power.
  Thus our result does not rule out the existence of an efficient algorithm for non-prime-power $m$. 
  Is there a polynomial-time algorithm computing $\zeta_m$-permanent where $m$ is not prime power?
  
  \item 
  First proposed by Aaronson and Arkhipov, the permanent anti-concentration conjecture is crucial for showing the equivalence of additive and multiplicative approximation of $1$-permanent for Gaussian matrices.

  Relation between the permanent anti-concentration conjecture and the $z$-permanent anti-concentration \conj{pacc}. Can one prove one from another or use polynomial interpolation to give multiple-to-one reduction?

  \item Higher moment of $|\Per_{z}(X)|^2$ when $X$ is drawn from complex Gaussian $\mathcal{G}$. 
\end{itemize}

\section*{Acknowledgement}
We thank Scott Aaronson, Mohammad Hafezi, and Dominik Hangleiter for useful discussions. EJ was supported by ARO W911NF-15-1-0397, National
Science Foundation QLCI grant OMA-2120757, AFOSRMURI FA9550-19-1-0399, Department of Energy QSA program, and the Simons Foundation.
SHH acknowledges the support from Simons Investigator in Computer Science award, award number 510817.

\bibliographystyle{alpha}
\bibliography{references}

\appendix

\section{Moments of $z$-permanents and Anti-Concentration Conjecture}\label{app:mom}

Anti-concentration not only plays an important rule in the Boson Sampling \cite{aaronson2011computational} but also has its own interests in random matrix theory \cite{eldar2018approximating,nezami2021permanent}. In this section, we discuss the moments of $z$-permanents drawn from i.i.d complex Gaussian and anti-concentration conjecture of $z$-permanent. 

First, we explicitly calculate the first moment of $|\Per_z(X)|^2$. Interestingly, the result does not depend on $z$.
\begin{lemma}\label{lem:first-m}
  For $z\in\mathbb T$, $\Exp_{X\sim\mathcal{G}}[|\Per_z(X)|^2] = n!$
\end{lemma}
\begin{proof}
  By direct calculation,
  \begin{align}\nonumber
    \Exp_{X\sim\G}[|\Per_z(X)|^2]
    &= \sum_{\sigma,\alpha} z^{\ell(\sigma)-\ell(\alpha)} \Exp_{X\sim\G}\left[\prod_{i}X_{i,\sigma(i)} X^*_{i,\alpha(i)} \right] \\\nonumber
    &= \sum_{\sigma,\alpha} z^{\ell(\sigma)-\ell(\alpha)} \prod_{i}\Id[\sigma(i)=\alpha(i)] \\\nonumber
    &= \sum_{\sigma,\alpha} z^{\ell(\sigma)-\ell(\alpha)} \cdot \Id[\sigma=\alpha] \\
    &= \sum_{\sigma} 1  = n!.
  \end{align}
\end{proof}
In fact, if we replace $\G$ with a random matrix whose each entry is independently sampled from any distribution of zero mean and unit variance, the second moment is $n!$.
Now we want to calculate $\Exp_{X\sim\G}[|\Per_z(X)|^4]$. 
We first recall the computation of $\Exp_{X\sim\G}[|\Per(X)|^4]$ in \cite{aaronson2011computational}.
\begin{lemma}
  For $z=1$, 
  $\Exp_{X\sim\G}[|\Per_z(X)|^4]=(n!)^2 (n+1)$.
\end{lemma}
\begin{proof}
  By direction calculation, 
  \begin{align}
    \Exp_{X\sim\G}[|\Per_z(X)|^4]
    &= \sum_{\sigma,\tau,\alpha,\beta} z^{\ell(\sigma)+\ell(\tau)-\ell(\alpha)-\ell(\beta)} \cdot 
      M(\sigma,\tau,\alpha,\beta),
      \end{align}
      where
      \begin{align}
      M(\sigma,\tau,\alpha,\beta) :=\prod_{i}\Exp_{X\sim\G}[X_{i,\sigma(i)}X_{i,\tau(i)}X_{i,\alpha(i)}^* X_{i,\beta(i)}^*].
  \end{align}
  If $\sigma\cup\tau\neq\alpha\cup\beta$, $M(\sigma,\tau,\alpha,\beta)=0$.

  To compute $M(\sigma,\tau,\alpha,\beta)$, we first provide some intuition: if $\tau=1$ and $\sigma=\gamma_1\ldots\gamma_m$, where each $\gamma_i$ is a cycle of length greater than 1. 
  Then either $\alpha$ has $\gamma_i$ and $\beta$ has the identity, or $\alpha$ has the identity and $\beta$ has $\gamma_i$.
  This implies that $\beta\alpha=\sigma$ and there are $2^m$ choices of $\alpha$.
  For arbitrary $\tau$, we can ``shift'' each permutation by multiplying $\tau^{-1}$ to the right and reducing to the above simpler case.
  Since $\sigma\tau^{-1}$ and $\tau\sigma^{-1}$ have the same cycle type, the above argument is symmetric with respect to swapping the roles of $\sigma$ and $\tau$ (i.e., we can start with $\sigma=1$ and get the same result).

  More formally, if $\sigma\tau^{-1}$ is of cycle type $\lambda=(\lambda_1,\ldots,\lambda_k)$ with $\lambda_1\geq\lambda_2\geq\ldots\geq \lambda_m> \lambda_{m+1}=\ldots=\lambda_k=1$, we may write in the form
  \begin{align}
    \sigma\tau^{-1} = \gamma_1\gamma_2\ldots\gamma_m.
  \end{align}
  Then if $\sigma\cup\tau=\alpha\cup\beta$, it holds that 
  \begin{align}
    (\alpha,\beta)\in A(\sigma,\tau) 
    \end{align}
    where
    \begin{align}
    A(\sigma,\tau) = \left\{ (\alpha,\beta): \alpha=\gamma_1^{b_1}\ldots\gamma_m^{b_m}\tau, \beta=\tau\alpha^{-1}\sigma, b_1,\ldots,b_m\in\bit\right\}.
  \end{align}
  Clearly for every $(\sigma,\tau)$ such that $\sigma\tau^{-1}$ is of type $\lambda$, $|A(\sigma,\tau)|=2^{m}$ since there are $2^m$ choice of $\alpha$ and $\beta$ is uniquely determined once $\alpha$ is fixed.

  Furthermore, for each fixed point $i$, $\sigma(i)=\tau(i)=\alpha(i)=\beta(i)$, and thus $\prod_{i:\sigma(i)=\tau(i)}\Exp_{X\sim\G}[|X_{i,\sigma(i)}|^4]=2^{k-m}$.
  Therefore, we have shown that 
  \begin{align}
    M(\sigma,\tau,\alpha,\beta) = 2^{F(\sigma\tau^{-1})}\cdot \Id[(\alpha,\beta)\in A(\sigma,\tau)]
  \end{align}
  and
  \begin{align}
    \sum_{\alpha,\beta} M(\sigma,\tau,\alpha,\beta) = 2^{C(\sigma\tau^{-1})},
  \end{align}
  where $F(\xi)$ is the number of fixed points, and $C(\xi)$ is the number of cycles in $\xi\in S_n$.  
  For $z=1$,
  \begin{align}
    \Exp_{X\sim\G}[|\Per(X)|^4]
    &=\sum_{\sigma,\tau,\alpha,\beta} M(\sigma,\tau,\alpha,\beta) = \sum_{\sigma,\tau} 2^{C(\sigma\tau^{-1})} = (n!)^2 (n+1).
  \end{align}
  Here we have applied the fact that $\sum_{\xi\in S_n}2^{C(\xi)}=(n+1)!$.
\end{proof}
However, we can not get the explicit formula for the $\Exp_{X\sim\G}[|\Per_z(X)|^4]$ because when $\sigma\cup\tau=\alpha\cup\beta$, the inversion number can be nonzero: namely $\ell(\sigma)+\ell(\tau)-\ell(\alpha)-\ell(\beta) \neq 0$. 
But we can still get the following result:
\begin{lemma}
  For $z\in\mathbb T$, $\Exp_{X\sim\G}[|\Per_z(X)|^4]\leq  \Exp_{X\sim\G}[|\Per(X)|^4].$
\end{lemma}
The proof is simple since $M(\sigma, \tau, \alpha, \beta) \geq 0$ and $\Exp_{X\sim\G}[|\Per_z(X)|^4]$ must be a positive number. So we have
\begin{align}\nonumber
    \Exp_{X\sim\G}[|\Per_z(X)|^4]
    &= \sum_{\sigma,\tau,\alpha,\beta} z^{\ell(\sigma)+\ell(\tau)-\ell(\alpha)-\ell(\beta)} \cdot 
      M(\sigma,\tau,\alpha,\beta) \\\nonumber
    &\leq \sum_{\sigma,\tau,\alpha,\beta}M(\sigma,\tau,\alpha,\beta) \\
    &=\Exp_{X\sim\G}[|\Per(X)|^4].
\end{align}
More generally, we have the following inequality.
\begin{lemma}
  For $z\in\mathbb T$ and $k\geq 2$, $\Exp_{X\sim\G}[|\Per_z(X)|^{2k}]\leq  \Exp_{X\sim\G}[|\Per(X)|^{2k}].$
\end{lemma}

Due to the lack of efficient formula for computing permanent, showing anti-concentration conjecture of permanent is hard. As we show in the paper,  There is no efficient formula for general $z$-permanent. So we believe that showing the anti-concentration conjecture of $z$-permanent is also a hard problem. However, the higher moment of $z$-permanent is less than permanent (but with the same first moment). This gives us a hint that anti-concentration is also true for $z$-permanent. We have the following similar as $z=1$.

\begin{theorem}[Weak Anti-Concentration of $z$-Permanent]\label{thm:weak-z-pacc}
For $\alpha<1,$
\begin{equation}
    \Pr_{X\sim\G}[|\Per_z(X)|^2 >\alpha \cdot n!] >\frac{(1-\alpha)^2}{n+1}.
\end{equation}
\end{theorem} 
\begin{proof}
By Chebyshev's inequality using
$\Exp_{X\sim\G}[|\Per_z(X)|^2]=n!$ and $\Exp_{X\sim\G}[|\Per_z(X)|^4]\leq n!^2(n+1)$, the bound immediately follows.
\end{proof}

\section{Strong Autocorrelation Property}\label{app:strong}
Consider a random variable $X$ with a real distribution $\mathcal{F}$ with mean $0$ and variance $1.$ We denote it as $X \sim \mathcal{F}_{\mathbb{R}}(0,1).$ We define $\mathcal{F}_{\mathbb{R}}(\mu,\sigma^2)$ with mean $\mu$ and variance $\sigma^2$ by shifting and scale the random variable $X \to \mu+\sigma^2 X.$ We denote such distribution $\F_{\mathbb{R}}(\mu, \sigma^2)$. We pick $\epsilon$ such that $\epsilon<\frac{1}{2}.$
\begin{align}
X & = \mathcal{F}_{\mathbb{R}}(0,1)^{n} \\ 
\mathcal{D}_{1}  &  =\mathcal{F}_{\mathbb{R}}\left(  0,\left(  1-\varepsilon\right)^2
\right)^{n},\\
\mathcal{D}_{2}  &  =\prod_{i=1}^{n}\mathcal{F}_{\mathbb{R}}\left(  v_{i},1\right)
\end{align}
We define
\begin{equation}
    G(\epsilon)=\left\Vert \mathcal{F}_{\mathbb{R}}(0,1)- \mathcal{F}_{\mathbb{R}}(0,(1-\epsilon)^2)
    \right\Vert, 
    H(v_i)=\left\Vert \mathcal{F}_{\mathbb{R}}(0,1)- \mathcal{F}_{\mathbb{R}}(v_i,1)
    \right\Vert.
\end{equation}

\begin{lemma}
If $G(\epsilon),\frac{dG(\epsilon)}{d\epsilon}$ are both continuous when $\epsilon \in [0,\frac{1}{2}]$ and $H(x), \frac{dH(x)}{dx}$ are both continuous when $x \in [0,1]$. Then,
\begin{align*}
     \left\Vert \mathcal{D}_{1}-X\right\Vert  &  \leq n M_1\epsilon, \forall \epsilon \in [0,1/2]\\ 
\left\Vert \mathcal{D}_{2}-X\right\Vert  &  \leq M_2\left\Vert 
v\right\Vert_{1},\,
\end{align*}
where $M_1,M_2$ are independent of $\epsilon,v_i$ for every $i$.
\end{lemma}

\begin{proof}
We first prove the first inequality:
\begin{align*}
    &\left\Vert \mathcal{D}_{1}-X\right\Vert \leq n \left\Vert \mathcal{F}_{\mathbb{R}}(0,1)- \mathcal{F}_{\mathbb{R}}(0,(1-\epsilon)^2)
    \right\Vert =n |G(\epsilon)| 
    \leq n\epsilon M_1.
\end{align*}
In order to bound the difference, we use the mean value theorem on variable $\epsilon$. 
Notice that $G(0)=0$. For all $\epsilon \in [0,\frac{1}{2}]$ we have:
\begin{equation}
   |G(\epsilon)-G(0)|\leq \epsilon \max_{\epsilon \in [0,\frac{1}{2}]} \left|\frac{dG(\epsilon)}{d\epsilon}\right|\leq \epsilon M_1.
\end{equation}
Here $M_1$ is defined to be $\left|\frac{dG(\epsilon)}{d\epsilon}\right|$. 
The existence is guaranteed due to the continuity of $\left|\frac{dG(\epsilon)}{d\epsilon}\right|$ and the following lemma.
\begin{lemma}[{Extreme value theorem~\cite{rudin1976principles}}]
 A continuous function from a non-empty compact space to a subset of the real numbers attains a maximum and a minimum.
\end{lemma}
So finish the proof of the first inequality.

The second inequality follows using again the triangle inequality:
\begin{align*}
    &\left\Vert \mathcal{D}_{2}-X\right\Vert \leq \sum_{i=1}^{n} \left\Vert \mathcal{F}_{\mathbb{R}}(0,1)- \mathcal{F}_{\mathbb{R}}(v_i,1)
    \right\Vert \leq \sum_{i=1}^{n} |v_i| M_2 \leq M_2\left\Vert 
v\right\Vert_{1}
\end{align*}
Here we use the same trick (notice $H(0)=0$, without loss of generality, we can assume $v_i \geq 0.$)
\begin{equation}
      |H(v_i)-H(0)|\leq |v_i| \max_{x \in [0,1]} \left|\frac{dH(x)}{dx}\right|\leq |v_i| M_2.
\end{equation}
Here $M_2$ is defined to be $\left|\frac{dH(x)}{dx}\right|$. 
\end{proof}
We can define complex random variable $Z=X+iY$,$X \sim \mathcal{F}_{\mathbb{R}}(\mu_1,1/2), Y \sim \mathcal{F}_{\mathbb{R}}(\mu_2,1/2)$ and $X,Y$ are independent. So then $Z \sim \mathcal{F}_{\mathbb{C}}(\mu_1+i\mu_2,1).$ 
It will be helpful to think of each complex coordinate as two real coordinates and $v$ is a vector in $\mathbb{R}^{2N}$.

\begin{corollary}
If we consider the complex version composed by real and with same condition:
\begin{align}
X & = \mathcal{F}_{\mathbb{C}}(0,1)^{n} = \mathcal{F}_{\mathbb{R}}(0,1/2)^{2n}\\ 
\mathcal{D}_{1}  &  =\mathcal{F}_{\mathbb{C}}\left(  0,\left(  1-\varepsilon\right)^2
\right)^{n}=\mathcal{F}_{\mathbb{R}}(0,(1-\epsilon)^2/2)^{2n},\\
\mathcal{D}_{2}  &  =\prod_{i=1}^{n}\mathcal{F}_{\mathbb{C}}\left(  v_{i},1\right)
\end{align}
We have:
\begin{align*}
     \left\Vert \mathcal{D}_{1}-X\right\Vert  &  \leq 2n M_1\epsilon, \forall \epsilon \in [0,\frac{1}{2}]\\ 
\left\Vert \mathcal{D}_{2}-X\right\Vert  &  \leq M_2\left\Vert 
v\right\Vert_{1} \leq \sqrt{2n}M_2\left\Vert 
v\right\Vert_{2}.
\end{align*}
In the final line, we use $1$- and $2$-norm inequalities. 
Notice that for the $1$-norm, we imagine $\nu$ as an $\mathbb{R}^{2n}$ vectors but for $2$-norm, we pack them into $\mathbb{C}^{n}$.  
\end{corollary}
Notice that in \cite{aaronson2011computational}, one can use rotational invariance to get the tighter bound for Gaussian matrices. But it is enough for our purpose.

\end{document}